\documentclass{statsoc}
\usepackage{amsmath}
\usepackage{amssymb}
\usepackage{bbm}
\usepackage{graphicx}
\usepackage[colorlinks=true, allcolors=blue]{hyperref}
\usepackage{doi}
\usepackage[a4paper,bottom=3.6cm]{geometry}
\usepackage{soul}

\newcommand{\expect}[1]{\mathbb{E}[#1]}
\newcommand{\var}{\text{Var}}
\newcommand{\pr}[1]{\Pr\left[#1\right]}
\newcommand{\x}{{\bf x}}

\newcommand{\simiid}{\overset{i.i.d.}{\sim}}

\newcommand{\errcv}{\text{e}_{\text{Kcv}}}
\newcommand{\errval}{\text{e}_{\text{val}}}
\newcommand{\errgen}{\text{e}_{\text{gen}}}
\newcommand{\Strain}{S_{\text{tr}}}
\newcommand{\Sval}{S_{\text{val}}}
\newcommand{\bias}{\text{bias}}
\newcommand{\hatT}{\widehat{T}}
\newcommand{\software}[1]{\texttt{#1}}
\newcommand{\R}{{\sf R}}
\newcommand{\TX}{\tilde{\mathcal X}}

\newcommand{\convprob}{\overset{P}{\to}}
\newcommand{\convdist}{\overset{\mathcal{D}}{\to}}

\DeclareMathOperator*{\argmin}{\arg\!\min}

\newcommand{\E}{\mathbb{E}}

\newtheorem{definition}{Definition}
\newtheorem{theorem}{Theorem}
\newtheorem{proposition}[theorem]{Proposition}
\newtheorem{corollary}[theorem]{Corollary}
\newtheorem{lemma}[theorem]{Lemma}
\newtheorem{remark}{Remark}

\title{On the cross-validation bias due to unsupervised pre\-processing}

\author{Amit Moscovich\footnote{Corresponding author}}
\address{Program in Applied and Computational Mathematics, Princeton University,\\Princeton NJ, USA.} 
\email{amit@moscovich.org}
\author[A. Moscovich and S. Rosset]{Saharon Rosset}
\address{Department of Statistics and Operations Research, Tel-Aviv University, Tel-Aviv, Israel.}
\email{saharon@post.tau.ac.il}

\begin{document}

\begin{abstract} 
Cross-validation is the de facto standard for predictive model evaluation and selection.
In proper use, it provides an unbiased estimate of a model's predictive performance. 
However, data sets often undergo various forms of data-dependent preprocessing,
such as mean-centering, rescaling, dimensionality reduction, and outlier removal.
It is often believed that such preprocessing stages, 
if done in an \textit{unsupervised} manner (that does not incorporate the class labels or response values) are generally safe to do prior to cross-validation.

In this paper, we study three commonly-practiced preprocessing procedures prior to a regression analysis:
(i) variance-based feature selection; (ii) grouping of rare categorical features; and (iii) feature rescaling.
We demonstrate that unsupervised preprocessing can, in fact, introduce a substantial bias into cross-validation estimates and potentially hurt model selection.
This bias may be either positive or negative and its exact magnitude depends on all the parameters of the problem in an intricate manner.
Further research is needed to understand the real-world impact of this bias across different application domains, particularly when dealing with small sample sizes and high-dimensional data.

\end{abstract}


\section{Introduction} \label{sec:intro}

Predictive modeling is a topic at the core of statistics and machine learning that is concerned with  predicting an output $y$ given an input
$\x$.
There are many well-established algorithms for constructing predictors $f:\mathcal X \to \mathcal Y$ from a representative data set of
input-output pairs \(\{ (\x_1, y_1), \ldots (\x_N, y_N) \} \).
Procedures for \emph{model evaluation} and \emph{model
selection} are used to estimate the performance of predictors  on new observations and to choose between them.
Commonly used procedures for model evaluation and selection include leave-one-out cross-validation,
K-fold cross-validation, and the simple train-validation split. 
In all of these procedures, the data set $S$
is partitioned into a training set $\Strain$ and a validation set $\Sval$.
Then a predictor, or set of predictors, is constructed from $\Strain$ and evaluated on $\Sval$.
See Chapter 5 of \cite{JamesWittenHastieTibshirani2014} for background on  cross-validation
and \cite{ArlotCelisse2010} for a mathematical survey.

Given a predictor and assuming that all of the observations $(\x_i, y_i)$ are independent and identically distributed, the mean error that a predictor
has on a validation set
is an unbiased estimate of that predictor's \emph{generalization error}, or \emph{risk}, defined as the expected error on a new sample.
In practice, however, data sets are often preprocessed by a data-dependent transformation
prior to model evaluation.
A simple example is mean-centering, whereby one first computes the empirical mean $\hat{\boldsymbol{\mu}}$ of the feature vectors (covariates) in the data
set and then maps each feature vector via $T_{\hat{\boldsymbol{\mu}}}(\x) = \x - \hat{\boldsymbol{\mu}}$.
After such a preprocessing stage, the transformed validation set no longer has
the same distribution as new $T_{\hat{\boldsymbol{\mu}}}$-transformed observations.
This is due to the dependency between the validation set and $\hat{\boldsymbol{\mu}}$.
Hence, the validation error is no longer guaranteed to be an unbiased estimate of the generalization error.
Put differently, by including a preliminary preprocessing stage, constructed from both the training and validation sets,
leakage of information from the validation set is introduced that may have an adverse effect on model evaluation \citep{KaufmanRossetPerlichStitelman2012}.
We consider two types of data-dependent transformations.
\paragraph{Unsupervised transformations:} $T:\mathcal X \to \mathcal X$ that are constructed only from $\x_1, \ldots \x_N$.
Common examples include mean-centering, standardization/rescaling, dimensionality reduction, outlier removal and grouping of categorical values.
\paragraph{Supervised transformations:} $T:\mathcal X \to \mathcal X$ whose construction depends on both  $\x_1, \ldots \x_N$ and $y_1, \ldots
y_N$.
Various forms of feature selection fall into this category.\\

Preliminary \textit{supervised} preprocessing is a well-known (but often repeated) pitfall.
For example,  performing feature selection on the entire data set
may find features that happen to work particularly well on the validation set, thus
typically leading to optimistic error estimates \citep{AmbroiseMclachlan2002, SimonEtal2003}.
In contrast, unsupervised preprocessing is widely practiced and believed by leading statisticians to be safe.
For example, in \textit{The Elements of Statistical Learning} \citep[p. 246]{HastieTibshiraniFriedman2009},
the authors warn against supervised preprocessing, but make the following claim regarding unsupervised preprocessing:

\begin{quote}
    \textit{
        In general, with a multistep modeling procedure, cross-validation must be applied to the entire sequence of modeling steps.
        In particular, samples must be ``left out'' before any selection or filtering steps are applied. There is one qualification: initial
unsupervised screening steps can be done before samples are left out.
        For example we could select the 1000 predictors with highest variance across all 50 samples,
        before starting cross-validation.
        Since this filtering does not involve the class labels,
        it does not give the predictors an unfair advantage.
    }
    \label{quote:ESL}
\end{quote}
In this paper, we show that contrary to these widely held beliefs, various forms of unsupervised preprocessing may, in fact, introduce a substantial bias to the cross-validation estimates of model performance.
Our main example, described in Section \ref{sec:feature_selected_linreg}, is an analysis of variance-based filtering prior to linear regression in the spirit of the setup in the quote above.
We demonstrate that the bias of cross-validation in this case can be large and  have an adverse effect on model selection.
Furthermore, we show that the sign and magnitude of the resulting bias depend on all the components of the data-processing pipeline: the distribution of the data points, the preprocessing transformation, the predictive procedure used, and the sizes of both the training and validation sets.

The rest of this paper proceeds as follows:
in Section \ref{sec:motivation} and the Supplementary we make the case that this methodological error is common in a wide range of scientific disciplines and tools, despite the fact that it can be easily avoided (see Section \ref{sec:therightway}).
In Sections \ref{sec:notation} and \ref{sec:basic_properties} we give the basic definitions and properties of the bias due to unsupervised preprocessing.
In addition to the main example presented in Section~\ref{sec:feature_selected_linreg}, in Section~\ref{sec:additional_examples} we study two additional examples:  grouping of rare categories and rescaling prior to lasso linear regression.
These examples shed more light on the origins of the bias and highlight the richness of the phenomenon.
In Section~\ref{sec:model_selection} we consider model selection in the presence of unsupervised preprocessing and demonstrate that it indeed induces a small performance penalty in the case of rescaled lasso linear regression.
In Section~\ref{sec:pathological}, we explain why meaningful upper bounds on the bias are unattainable under the most general
settings. However, we describe particular settings where such upper bounds may be attained after all.

\subsection{Motivation} \label{sec:motivation}
In this section, we establish the fact that unsupervised preprocessing is very common in science and engineering.
We first note that in some scientific fields, the standard methodology incorporates unsupervised preprocessing stages into the computational
pipelines.
For example in Genome-Wide Association Studies (GWAS) it is common to standardize genotypes to have zero mean and unit variance prior to analysis
\citep{YangEtal2010,SpeedBalding2014}.
In EEG studies, data sets are often preprocessed using independent component analysis, principal component analysis or similar methods
to remove artifacts such as those resulting from eye blinks \citep{UriguenGarciaZapirain2015}. 

To quantitatively estimate the prevalence of unsupervised preprocessing in scientific research,
we have conducted a review of research articles published in \textit{Science Magazine} over a period of 1.5 years.
During this period, we identified a total of 20 publications that employ cross-validated predictive modeling.
After carefully reading them, we conclude that seven of those papers (35\%) performed some kind of unsupervised preprocessing \emph{on the entire data set} prior to cross-validation. Specifically, three papers filtered a categorical feature based on its count \citep{Dakin2018, Cohen2018, Scheib2018};
two papers performed feature standardization \citep{Liu2017, Ahneman2018}; one paper discretized a continuous variable, with cutoffs based on
its percentiles \citep{Davoli2017};
and one paper computed PCA on the entire data set, and then projected the data onto the first principal axis \citep{Ni2018}.
The full details of our review appear in the Supplementary.

Many practitioners are careful to always split the data set into a training set and a validation set before any processing is performed.
However, it is often the case, both in academia and industry, that by the time the data is received it has already undergone various stages of
preprocessing.
Furthermore, even in the optimal case, when the raw data is available, some of the standard software tools do not currently have the built-in facilities to correctly incorporate preprocessing into the cross-validation procedure.
One example is the widely used \software{LIBSVM} package \citep{libsvm2011}.
In their user guide, they recommend to first scale all of the features in the entire data set using the \software{svm-scale} command and only
then to perform cross-validation or train-validation splitting \citep{libsvmguide2010}. Furthermore, it seems there is no easy way to use their
 command line tools to perform scaling that is based only on the training set and then apply it to the validation set.

\subsection{The right way to combine preprocessing and cross-validation} \label{sec:therightway}

To guarantee that the cross-validation estimator is an unbiased estimator of model performance, all data-dependent unsupervised preprocessing
operations should be determined using only the training set $\Strain$ and then merely applied to the validation set $\Sval$, as is commonly done for (label-dependent) feature-selection and other supervised preprocessing procedures.
In other words, preprocessing steps should be deemed an inseparable part of the learning algorithm.
We summarize this approach here:

\paragraph{(Step 1) Preprocessing.}
Learn a feature transformation $\hatT: \mathcal{X} \to \TX$ using just the training set $\Strain$.

\paragraph{(Step 2) Training.}
Transform the feature vectors of $\Strain$ using  $\hatT$ and then learn a predictor $\hat{f}_{\Strain}$ from the transformed
training set
$\left\{ \big(\hatT(\x), y \big) : (\x,y) \in \Strain \right\}$.

\paragraph{(Step 3) Validation.} For every observation $(\x, y)$ in $\Sval$, compute a prediction for the transformed feature
vector  $\hat{y} = \hat{f}_{\Strain}\big(\hatT(\x)\big)$
and evaluate some loss function $\ell(y, \hat{y})$.\\\\
For example, to perform standardization of univariate data one would estimate the empirical mean $\hat{\mu}_{\text{tr}}$ and empirical
standard deviation
$\hat{\sigma}_{\text{tr}}$ of the covariates in $\Strain$ and then construct the standardizing transformation $\hatT(x) = (x-\hat{\mu}_{\text{tr}})/\hat{\sigma}_{\text{tr}}$ to be applied to both the training and validation sets.
In a cross-validation procedure, the above steps would be repeated  for every
split of the data set into a training set and a validation set.

Some of the leading frameworks for predictive modeling provide mechanisms to effortlessly perform the above steps.
Examples include the \software{preProcess} function of the caret \R\ package,
the pipeline module in the scikit-learn Python library
and ML Pipelines in Apache Spark MLlib \citep{Kuhn2008, scikitlearn2011, MLlib2016}.

\subsection{Related works}

To the best of our knowledge, the only previous work that directly addresses biases due to unsupervised preprocessing is an empirical study
focused on gene microarray analysis \citep{Hornung2015}.
They studied several forms of preprocessing, including variance-based filtering and imputation of missing values, but mainly focused on PCA dimensionality reduction and Robust Multi-Array Averaging (RMA), a multi-step data normalization procedure for gene microarrays.
In contrast to our work, they do not measure the bias of the cross-validation error with respect to the generalization error of the same model.
Rather, they consider the ratio of cross-validation errors of two different models: one where the preprocessing is performed on the entire data
set, and one where the preprocessing is done in the ``proper'' way, as in the three step procedure outlined in Section \ref{sec:therightway}.
They conclude that RMA does not incur a substantial bias in their experiments, whereas PCA dimensionality reduction results in overly-optimistic
cross-validation estimates.

\section{Notation and definitions} \label{sec:notation}

In this section, we define some basic notation and use it to express two related approaches for model evaluation: train-validation splitting and cross-validation.
Then we define the bias due to unsupervised preprocessing.

\subsection{Statistical learning theory and cross-validation}
Let $\mathcal X$ be an input space, $\mathcal Y$ be an output space and $\mathcal D$ a probability distribution over the space $\mathcal X \times
\mathcal
Y$ of input-output pairs.
Let $S = \{(\x_1, y_1), \ldots, (\x_N, y_N)\}$ be a set of pairs sampled independently from $\mathcal D$.
In the basic paradigm of statistical learning, an algorithm $A$ for learning predictors
takes $S$ as input and outputs a predictor $\hat{f}_S:\mathcal X \to \mathcal Y$.
To evaluate predictors' performance, we need a loss function $\ell(y,y')$
that quantifies the penalty of predicting $y'$ when the true value is $y$.
The \emph{generalization error} (or \emph{risk}) of a predictor $f$ is its expected loss, 
\begin{align}
    \errgen(f, \mathcal D) := \E_{\x, y} \, \ell \big( y, f(\x) \big).
\end{align}
When the distribution $\mathcal D$ is known, the generalization error can be estimated directly by integration or repeated sampling. However,
in many cases, we only have a finite set of $N$ observations at our disposal. In that case, a common approach
for estimating the performance of a modeling algorithm is to split the data set into a \emph{training set} of size $n$ and a \emph{validation
set} of size $m=N-n$.
\begin{align}
    \Strain &= \{(\x_1, y_1), \ldots, (\x_n, y_n)\} \\
    \Sval &= \{(\x_{n+1}, y_{n+1}), \ldots, (\x_{n+m}, y_{n+m})\} \qquad(N = n+m)
\end{align}
The model is then constructed (or trained) on the training set and evaluated on the validation set.
We denote the learned predictor by $\hat{f}_{\Strain}$.
Its validation error is the average loss over the validation set,
\begin{align}
    \errval(\hat{f}_{\Strain}, S_{\text{val}} ) := \frac{1}{|\Sval|} \sum_{(\x,y) \in \Sval} \ell \big( y, \hat{f}_{\Strain}(\x) \big).
\end{align}
This approach is known as the \emph{train-validation split} (or train-test split).
Its key property is that it provides an unbiased estimate of the predictor's generalization error
given a training set of size $n$, since we have, for any predictor $f$,
\begin{align}
    \E_{\Sval} \, 
    \errval(f, \Sval)
    =
    \errgen(f, \mathcal D)
    .
\end{align}
A more sophisticated approach is \emph{K-fold cross-validation}.
In this approach the data set $S$ is partitioned into $K$ folds of size $N/K$
(we assume for simplicity that $N$ is divisible by $K$).
The model is then trained on $K-1$ folds and its average loss is computed
on the remaining fold. This is repeated for all $K$ choices of the validation
fold and the results are averaged to form the \emph{K-fold cross-validation error} $\errcv$.

\subsection{The bias due to unsupervised preprocessing} \label{sec:biasdef}

We study the setting where the instances of both the training and validation sets undergo an unsupervised transformation prior
to cross-validation.
We denote by
\begin{align}
    A_T: \mathcal X^{n+m} \to (\mathcal X \to \TX),
\end{align}
an unsupervised procedure that takes as input the set of feature vectors $\{\x_1, \ldots, \x_{n+m}\}$ and outputs a transformation $T: \mathcal
X \to \TX$.
The space of transformed feature vectors $\tilde{\mathcal X}$ may be equal to $\mathcal X$, for example when $T$ is a scaling transformation,
or it may be different, for example when $T$ is some form of dimensionality reduction.
We denote by
\begin{align}
    A_f: (\TX \times \mathcal Y)^n \to (\TX \to \mathcal Y),
\end{align}
a learning algorithm that takes  a transformed training set $\{ \big( T(\x_1), y_1 \big), \ldots, \big( T(\x_n), y_n \big)\}$ and outputs a predictor for transformed
feature vectors $f: \TX \to \mathcal Y$.
In the following we denote \( \hatT := \hatT_S = A_T(\x_1, \ldots, \x_{n+m}) \)
and \( \hat{f} := \hat{f}_S = A_f\left\{ \big( \hatT(\x_1), y_1\big), \ldots, \big( \hatT(\x_n), y_n \big) \right\}. \)
The validation error is
\begin{align} \label{def:errval}
    \errval(A_T, A_f, S) := \frac{1}{|\Sval|} \sum_{(\x,y) \in \Sval} \ell[y, \hat{f}\{\hatT(\x)\}].
\end{align}
Likewise, the generalization error is
\begin{align} \label{def:errgen}
    \errgen(A_T, A_f, S, \mathcal D) := \E_{(\x,y) \sim \mathcal D} \, \ell [ y, \hat{f}\{\hatT(\x)\}] .
\end{align}
The focus of this paper is the bias of the validation error with respect to the generalization error, due to the
fact that the feature vectors in the validation set were involved in forming the unsupervised transformation $\hatT$.
\begin{definition} \label{def:biasp}
    The bias of a learning procedure $(A_T, A_f)$ composed of an unsupervised transformation-learning algorithm $A_T$ and a predictor learning algorithm $A_f$ is 
    \begin{align} \label{eq:ppbias}
        \bias(A_T, A_f, \mathcal D, n, m)
        &:=
        \E
        \left\{
            \errval(A_T, A_f, S)
            -
            \errgen(A_T, A_f, S, \mathcal D)
        \right\}.
    \end{align}
\end{definition}
Note that instead of analyzing the bias of the train-validation split estimator, we may consider the bias of K-fold cross-validation
\(
    \expect{\errcv-\errgen}
\).
However, due to the linearity of expectation, this bias is equal to $\bias(A_T, A_f, \mathcal D,(K-1)s,s)$ where $s$ is the
fold size.
Hence, our analysis applies equally well to K-fold cross-validation.

\section{Basic properties of the bias} \label{sec:basic_properties}

Practically all methods of preprocessing learned from i.i.d. data do not depend on the order of their inputs.
Thus, typically $A_T$ is a symmetric function.
This simplifies  the expression for the expected bias.

\begin{proposition}
    If $A_T$ is a symmetric function then the expected validation error admits the following simplified form,
    \begin{align} \label{eq:expected_eval_symmetric}
        \E_{S} \, \errval = \E_{S} \, \ell [ y_{n+1}, \hat{f}\{ \hatT(\x_{n+1}) \} ].
    \end{align}
    Hence we obtain a simplified expression for the bias,
    \begin{align} \label{eq:bias_symm_T}
        \bias(A_T, A_f, \mathcal D, n, m)
        =
        \E_{S, \x,y} \, \ell [ y_{n+1}, \hat{f}\{ \hatT(\x_{n+1}) \} ] - \ell [ y, \hat{f} \{ \hatT(\x) \} ].
    \end{align}
\end{proposition}
\begin{proof}
If $A_T$ is invariant to permutations of its input then the vector of transformed training feature vectors 
$(\hatT(\x_1), \ldots, \hatT(\x_n))$ is invariant to permutations of the feature vectors in the validation set $\x_{n+1}, \ldots, \x_{n+m}$.
Hence the chosen predictor $\hat{f}$ does not depend on the ordering of the validation covariates.
It follows that the random variables $\ell(y_{i}, \hat{f}(\hatT(\x_{i})))$ are identically distributed for all $i \in \{n+1, \ldots, n+m\}$.
Eq. \eqref{eq:expected_eval_symmetric} follows from \eqref{def:errval} by the linearity of expectation.
$\hspace*{\fill} \square$
\end{proof}

\begin{remark}
    Even though the expected validation error in \eqref{eq:expected_eval_symmetric} does not explicitly depend on $n,m$, there is an implicit
dependence
    due to the fact that the distributions of the selected transformation $\hatT$ and predictor $\hat{f}$ depend on $n$ and $m$.
\end{remark}

Were the feature transformations chosen in a manner that is data-independent, then the transformed validation covariates $\hatT(\x_{n+1}), \ldots,
\hatT(\x_{n+m})$
would be a independent and distributed as $\hatT(\x)$ where $\x \sim \mathcal D_{\mathcal X}$.
In that case, $\bias$ would be zero. However, since $\hatT$ is chosen in a manner that depends on $\x_{n+1}, \ldots, \x_{n+m}$,
the distribution of $\hatT(\x_i)$ for $i \in \{n+1, \ldots, n+m\}$ may be vastly different from that of $\hatT(\x)$ for newly generated observations.
For an extreme example of this phenomenon see Section~\ref{sec:pathological}.

\section{Main example: feature selection for high-dimensional linear regression} \label{sec:feature_selected_linreg}

In this section we consider variance-based feature selection performed on the entire data set prior to linear regression and evaluated using cross-validation.
We demonstrate that this can incur a substantial bias in the validation error with respect to the true model error.
This is demonstrated on both synthetic data and on a real dataset. 
The details of our experiment follow.

\paragraph{Sampling distribution:}
We generate a random vector of coefficients $\boldsymbol{\beta} = (\beta_1, \ldots, \beta_p)^T$ where $\beta_i \sim \mathcal{N}(0,1)$.
Each observation $(\x, y)$ is given by:
\begin{align}
    \x = (C x_1, \ldots, C x_M, x_{M+1}, \ldots, x_p)  \qquad y = \x \boldsymbol{\beta} +\ \epsilon
\end{align}
where $x_1, \ldots, x_p$ are drawn i.i.d. from some zero-mean distribution, $C>1$ is a constant, and $\epsilon$ is a Gaussian noise term. We have tested two distributions for $x_1, \ldots, x_p$: a standard Gaussian and a t-distribution with 4 degrees of freedom.
By construction, the variance of $\x \boldsymbol{\beta}$ is proportional to $(p-M)+ C^2M$, so given a noise level $\eta > 0$ we set the noise term to
\(
    \epsilon \sim \mathcal{N}(0, \sigma^2).
\)
where $\sigma^2 = \eta \cdot\{ (p-M)+C^2M \}$.
                                                                                                                                                                                                                                                                                                                        
\paragraph{Preprocessing:}
Variance-based feature selection. The unsupervised transformation is $\widehat{T}(\x) = (x_{j_1}, \ldots, x_{j_K})$ where $j_1, \ldots, j_K$ are the $K$ covariates with highest empirical variance. Importantly, this variance is computed over the entire dataset.

\paragraph{Predictor:}
Linear regression with no intercept: $\hat{f}\{\hatT(\x)\} = (x_{j_1}, \ldots, x_{j_K})\hat{\boldsymbol{\beta}}$, where
\begin{align}
    \hat{\boldsymbol{\beta}}
    &:=
    \argmin_{\boldsymbol{\beta} \in \mathbb{R}^K} \sum_{i=1}^n \left( \beta_1 x_{j_1} + \ldots + \beta_K x_{j_K} - y_i\right)^2 .
\end{align}

We've chosen the sampling distribution so that the first $M$ features will have a larger magnitude and a correspondingly larger influence on the response.

\begin{figure}
    \includegraphics[width=0.5\linewidth]{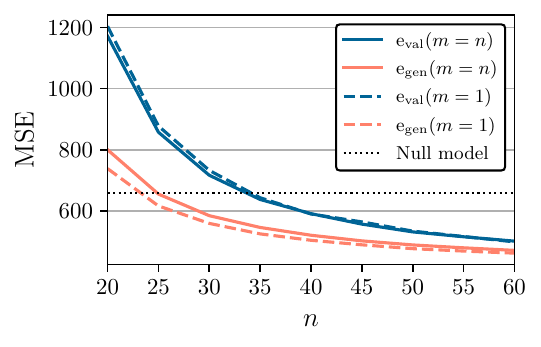}
    \includegraphics[width=0.5\linewidth]{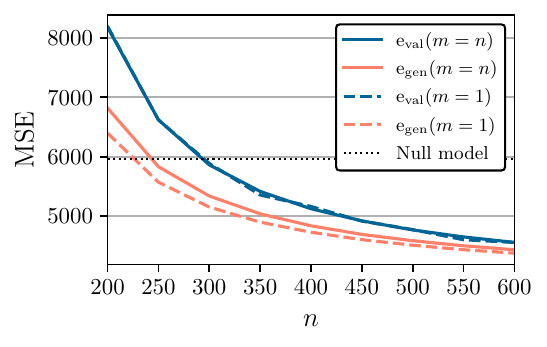}
    \includegraphics[width=0.5\linewidth]{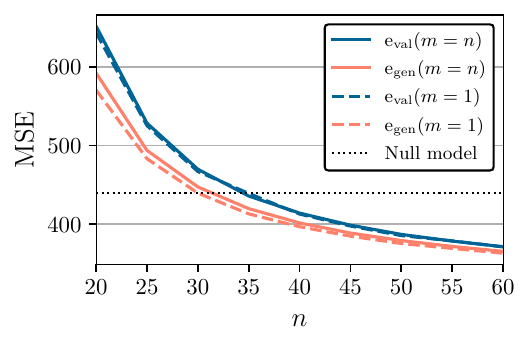}
    \includegraphics[width=0.5\linewidth]{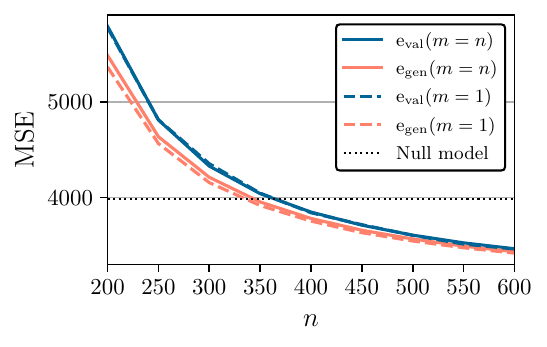}
    \caption{Validation and generalization mean squared errors (MSE) for variance-based feature selection followed by high-dimensional
linear regression.
$n$ is the size of the training set and $m$ is the size of the validation set.
    Solid lines correspond to $m=n$ as in two-fold cross-validation whereas dashed lines correspond to
a $m=1$ as in leave-one-out cross-validation. The black dotted line is the error of  the null
model $f(\x) \equiv 0$.
    The first $M$ variables are multiplied by the constant $C>1$. The feature selection procedure  picks the $K$ columns
with highest variance over the entire sample of size $m+n$.}
    \label{fig:variance_selected_linear_regression}
    \sffamily\footnotesize
    \begin{tabular}{lllll}
        (top left)     & $p=100$  & $M=C=5$  & $K=10$  & $x_{ij} \sim t(4)$. \\
        (bottom left)  & $p=100$  & $M=C=5$ & $K=10$  & $x_{ij} \sim \mathcal{N}(0,1)$. \\
        (top right)    & $p=1000$ & $M=C=10$ & $K=100$ & $x_{ij} \sim t(4)$. \\
        (bottom right) & $p=1000$ & $M=C=10$ & $K=100$ & $x_{ij} \sim \mathcal{N}(0,1)$.
    \end{tabular}
\end{figure}

\begin{figure}
    \includegraphics[width=0.5\linewidth]{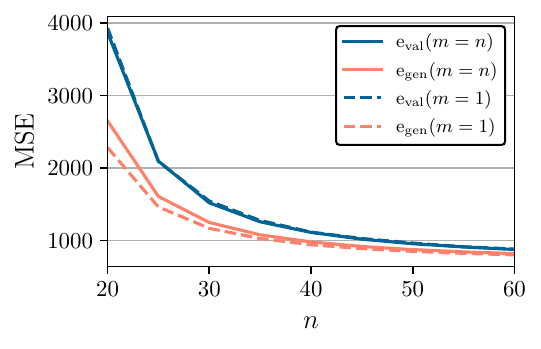}
    \includegraphics[width=0.5\linewidth]{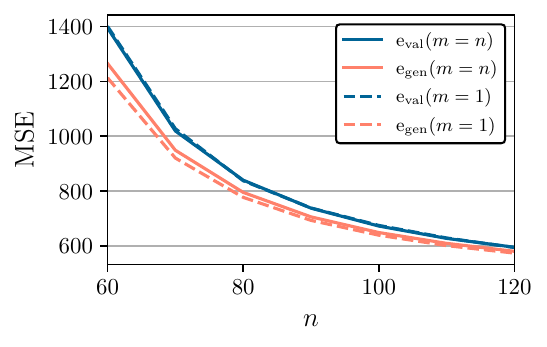}
    \caption{Validation and generalization errors for data sampled from the {\bf superconductivity} dataset, averaged over
one million repetitions. $n=21263$,\ \  $p=82$. (left) 
$K=10$; (right) $K=30$.}
    \label{fig:realdatasets}
\end{figure}

\subsection{Simulation study}

In Figure~\ref{fig:variance_selected_linear_regression} we compare the validation and generalization error of the above model for several choices of the parameters.
These results were obtained by a simulation and averaged over 100,000 runs.
Error bars were omitted since the uncertainty is negligible. 
Here we show the results for a moderate level of noise $\eta=1$.
We also tested other noise levels, however, these experiments resulted in qualitatively similar plots and so we chose to omit them.
We end this section with several remarks regarding the figures above:
\begin{remark}[sign of the bias]
    In these experiments, the validation error is larger than the generalization error, hence the bias is positive.
    While this may seem counter-intuitive, it is in fact a consequence of the excess variance in the validation error, which results from selecting the high-variance features on the entire data. See our analytical derivation for a simplified setting in the next subsection. 
    In general, the bias can be either negative or positive.
    In Section \ref{sec:categorical} we even show an example where the bias flips sign as the noise level is increased.
\end{remark}
\begin{remark}[asymptotics]
    In Figure~\ref{fig:variance_selected_linear_regression}, the bias vanishes in the regime where $p$ is fixed and $n \to \infty$.
    In the next subsection, we  prove this claim rigorously for a simplified model of variance-based feature selection followed by linear regression.
    However, this is not true for all preprocessing procedures as we later show in Section~\ref{sec:pathological}.
\end{remark}

\subsection{Experiments on a real dataset}

In addition to synthetic data, in Figure \ref{fig:realdatasets} we show results on a real dataset, the superconductivity
dataset from the UCI repository \citep{uci}.
This dataset contains $82$ chemical properties for a large collection of superconductors, e.g. their atomic mass, thermal
conductivity, etc.
The goal is to predict the critical temperature at which a superconductor becomes superconductive \citep{Hamidieh2018}.
The sampling distribution used in our simulations is defined by taking random subsets of chemicals from the standardized
superconductivity training set without replacement.

\begin{remark}[stability of the preprocessing procedure]
    It may be the case that, on a given dataset, the scale differences between the features are so large that the same set
of high-variance features is selected almost every single time. In that case, the feature-selecting transformation $\hatT$
is almost independent of the dataset. Therefore, we can expect the bias to be close to zero (see Section \ref{sec:basic_properties}).
This is indeed the case on the unstandardized superconductivity dataset. For example, selecting a random subset of 20 chemicals
and then picking the 10 features with highest variance yielded the exact same set of features 995 times out of 1,000 runs.
\end{remark}

\subsection{Analysis} \label{sec:feature_selected_linreg_analysis}

To understand where the bias is coming from, we consider the simple noiseless setting with i.i.d. Gaussian covariates and a single selected variable.
In our notation this means $K=1, M=\eta=0$.
Figure \ref{fig:variable_selected_regression_K1} compares the average validation and generalization errors for $p=50$ features and training set sizes $n=5,10,\ldots,40$.
We see that even this stripped-down model shows a large gap between the validation and generalization error.

\begin{figure}
    \includegraphics[width=0.5\linewidth]{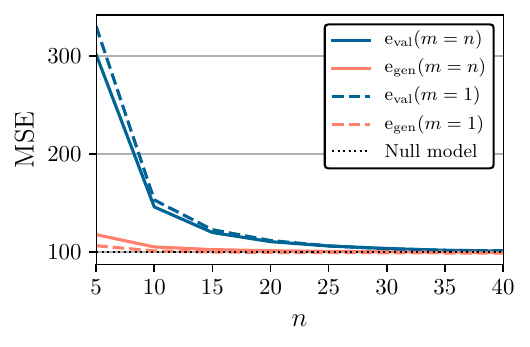}
    \includegraphics[width=0.5\linewidth]{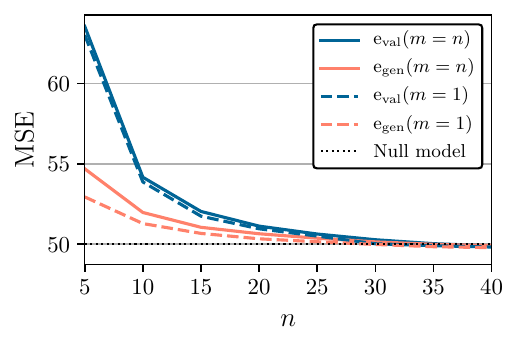}
    \caption{Validation and generalization errors for the theoretically analyzed setup in Section \ref{sec:feature_selected_linreg_analysis},
of selecting the single feature with largest sum of squares out of $p=50$ features, followed by linear regression.
    (left) $x_{ij} \sim t(4)$; (right) $x_{ij} \sim \mathcal{N}(0,1)$.} \label{fig:variable_selected_regression_K1}
\end{figure}

To simplify the analysis, rather than variance-based variable selection, we consider selecting the variable with the largest squared norm (sum of squares).
Since the covariates are drawn from a zero-mean distribution, the averaged sum of squares is a consistent estimator of the variance.
While not shown here, the plots in Figure~\ref{fig:variable_selected_regression_K1} for both variable-based and norm-based feature filtering appear indistinguishable. 
Let $X \in \mathbb{R}^{(n+m) \times p}$ be the matrix of feature vectors, where the first $n$ rows are the feature vectors of the training sample and the rows $n+1, \ldots, n+m$ contain the feature vectors of the validation sample.
Let $X_{*j} = (X_{1,j}, \ldots, X_{n,j})^T$ be the j\textsuperscript{th} feature across all training samples.
We denote the normalized dot-product (correlation) between the j\textsuperscript{th} and k\textsuperscript{th} training feature
\begin{align}
    \hat{\rho}_{jk} = \frac{X_{*j}^T X_{*k}}{\|X_{*j}\|\|X_{*k}\|}.
\end{align}
We now derive an expression for the bias due to preliminary variable selection.
\begin{theorem} \label{prop:var_selected_linreg}
    Let $\widehat{j}$ be the maximizer of $\sum_{i=1}^{n+m} X_{i,j}^2$ and let \(j_o\) be any other column (they are exchangeable).
    For the model described above with $K=1, M=\eta=0$,
    the bias of the MSE due to the feature selection is
   \begin{align} \label{eq:bias_var_selected_linreg}
        \bias
        =
        \E \left[ \left\{ (p-1) \hat{\rho}_{\widehat{j}j_o}^2 \frac{\|X_{*,j_o}\|^2}{\|X_{*, \widehat{j}}\|^2} - 1\right\}
     \left( X_{n+1,\widehat{j}}^2 - 1 \right) \right].
    \end{align}
\end{theorem}
The proofs of the theorem and following corollary are in the appendix. 
\begin{corollary} \label{cor:var_selected_linreg}
    From Eq. \eqref{eq:bias_var_selected_linreg} we can infer the following asymptotic results:
    \begin{enumerate}
        \item In the  classical regime where $p$ is fixed and $n \to \infty$, it follows  that \( \bias \to 0 \).
        \item If we fix $n,m$ and take $p \to \infty$ then \(\frac{\bias}{p/n} \to 1 \) and in particular $\bias \to \infty$.
    \end{enumerate}
\end{corollary}

\section{Additional examples} \label{sec:additional_examples}
\subsection{Grouping of rare categories} \label{sec:categorical}

Categorical covariates are common in many real-world prediction tasks.
Such covariates often have long-tailed distributions with many rare categories. This creates a problem since there is no
way to accurately estimate
the responses associated with the rare categories.
One common solution to this problem is to preprocess the data by grouping categories that have only a few observations
into a \emph{rare}
category.
See for example \citep[Section 4.1]{Harrell2015} and \citep[Section 15.5]{WickhamGrolemund2017}.

In this section, we analyze this type of grouping prior to a  simple regression problem of estimating a response given only the category of the observation.
We show that if the grouping of rare categories is done on the entire data set before the train-validation split then the validation error is biased
with respect to
the generalization error.

\paragraph{Sampling distribution:}
We draw mean category responses $\mu_1, \ldots, \mu_C$ independently from $\mathcal{N}(0,1)$.
To generate an observation $(x,y)$ we  draw $x$ uniformly from  $\{1, \ldots, C\}$
and then set $y$ to be a noisy measurement of that category's mean response.
\begin{align}
    x \sim \mathcal U\{1, \ldots, C\},
    \qquad
    y \sim N(\mu_x, \sigma^2).
\end{align}

\paragraph{Preprocessing:}
We group  all categories that appear less than $M$ times in the union of the training and validation sets into a \textit{rare} category.

\paragraph{Predictor:}
Let $Y(k) := \left\{ y_i : x_i = k \text{ for } i=1,\ldots,n \right\}$
be the set of sampled responses for category $k$ in the training set. The predicted response is
\begin{align}
    \hat{f}(k) = 
    \begin{cases}
        \text{mean}\{Y(k)\} & \text{if } |Y(k)| \ge 1 \text{ and } k \text{ is not } rare \\
        0 & \text{otherwise}.
    \end{cases}
\end{align}
To simplify the analysis, we choose to set the 
estimated response of the rare category to zero, rather than to the mean of its responses,
which is zero in expectation.

\subsubsection{Analysis}

Let $x_1, \ldots, x_{n+m}$ be the categories in our sample
where the first $n$ belong to the training set and the rest to the validation set.
Denote by $\#_q(k) = \sum_{i=1}^q \mathbbm{1}(x_i=k)$
the number of appearances of a category $k$ among the first $q$ observations in $x_1, \ldots, x_{n+m}$.
Denote by $r_k$ the event that $\#_{n+m}(k) < M$, i.e. that the category k is determined \emph{rare},
and by
\begin{align}
    p_i(k) := \pr{\#_n(k)=i|\neg r_k}
\end{align}
the probability of having exactly
$i$ observations of category $k$ in the training set, given that this category is not rare.
\begin{theorem} \label{prop:rare_categories_MSE}
    The Mean Squared Error (MSE) for an estimate of an observation from category $k$ with category mean $\mu_k$ is 
    \begin{align} \label{eq:MSE_k}
        \sigma^2
        +
        \pr{r_k} \mu_k^2 + \pr{\neg r_k} p_0(k) \mu_k^2
        +
        \sigma^2 \pr{\neg r_k} \sum_{i=1}^{n} \frac{p_i(k)}{i}.
    \end{align}
\end{theorem}
See the appendix for the proof.
At first, it may seem that the MSE should be the same for samples in the validation
set as for newly generated samples.
However, the probabilities $\pr{r_k}$ and $p_0(k), \ldots, p_{n}(k)$ are different in the two cases.
This stems from the fact that whenever we consider an observation from the validation set, we are guaranteed
that its category appears at least once in the data set.
For example, consider the case of an $m=1$ sized validation set, as in leave-one-out cross-validation, and
let the rare category cut-off be $M=2$.
Note that $p_0(k) = 0$ in this case, hence the third term of \eqref{eq:MSE_k} vanishes.
We are left with the following MSE for an observation of category $k$,
\begin{align} \label{eq:MSE_k_M2}
    \sigma^2 + \pr{r_k} \mu_k^2 + \sigma^2 \pr{\neg r_k} \sum_{i=1}^{n} \frac{p_i(k)}{i}.
\end{align}
Since the validation set contains a single observation $(x_{n+1}, y_{n+1})$ and since $M=2$,
given that $x_{n+1}=k$,
the category $k$ will be considered \emph{rare} if and only if the training set contains exactly zero observations of it.
Hence,
\begin{align} \label{eq:pr_r_x_i}
    \pr{r_{x_{n+1}}} = \pr{\text{Bin}(n,1/C) = 0} = \left( 1 - \frac{1}{C} \right)^n.
\end{align}
In contrast, the category of a newly generated observation $(x,y)$ will be \emph{rare} if and only if the training and
validation sets contain
\emph{zero or one} observations of it. Hence,
\begin{align} \label{eq:pr_r_x}
    \pr{r_{x}} = \pr{\text{Bin}(n+1,1/C) \le 1} = \left(1 - \frac{1}{C} \right)^{n+1} + \frac{n+1}{C} \left(1 - \frac{1}{C} \right)^n.
\end{align}
Let us denote
\(
    A(k) = \pr{\neg r_k} \sum_{i=1}^{n} \frac{p_i(k)}{i}.
\)
Since the categories are drawn uniformly, $A(k)$ does not depend on the specific category $k$,
but does depend on whether or not it is in the validation set. 
Since $\expect{\mu_k^2} = 1$,
it follows from \eqref{eq:MSE_k_M2} that for $M=2$,
the bias as given by Eq. \eqref{eq:bias_symm_T} satisfies
\begin{align}
        \bias
        &=
        \E \left(
            \ell [y_{n+1}, \hat{f}\{\hatT(x_{n+1})\}]
            -
            \ell[y, \hat{f}\{\hatT(x)\}]
        \right) \\
        &=
        \pr{r_{x_{n+1}}} - \pr{r_{x}}
        +
        \sigma^2
        \left(
            A(x_{n+1})
            -
            A(x)
        \right).
\end{align}
For the noiseless case $\sigma = 0$, we obtain
\begin{align}
    \bias = \pr{r_{x_{n+1}}} - \pr{r_{x}} = -\frac{n}{C} \left( 1 - \frac{1}{C} \right)^n \approx -\frac{n}{C} \exp(-n/C).
\end{align}
We see that in the noiseless setting, with the rare cutoff at $M=2$ and using leave-one-out cross-validation, the bias is
 always negative. Note that even in this simple case the bias is not a monotone function of the training set size.
Note also that the bias vanishes as $n/C \to \infty$. This is expected since in this regime there are no
rare categories.

\begin{figure}[t]
    \includegraphics[width=0.5\linewidth]{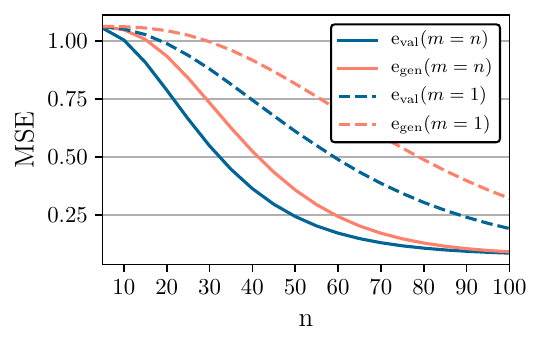}
    \includegraphics[width=0.5\linewidth]{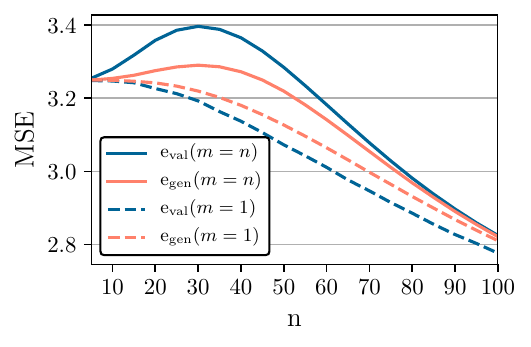}
    \caption{Validation and generalization errors of preliminary grouping with rare categories. Points in these plots are
the average of 10,000,000 runs. The number of categories is $C=20$ and the rare category cutoff is $M=4$. Error bars were
omitted as the uncertainty is negligible. $n$ is the number of training samples. Solid lines correspond to a validation set of size $m = n$. Dashed lines correspond to a validation set of size one. (left panel) low noise $\sigma = 0.25$; (right panel)
high noise  $\sigma = 1.5$.}
    \label{fig:categorical}
\end{figure}

\subsubsection{Simulation study}
Larger values of the category cutoff $M$ are more cumbersome to analyze mathematically but can be easily handled via simulation.
We present one such simulation for $C=20$ categories and $M=4$ in Figure \ref{fig:categorical}, where the empirical average
of $\errval$ and
$\errgen$ are plotted
for various training set sizes $n=5,10,\ldots,100$,  once with $m=n$ validation samples and once with $m=1$ validation samples.
Note that the bias, which corresponds to the difference between the blue and red lines, may be either negative or positive,
depending on the noise level and the validation set size.

\subsection{Rescaling prior to Lasso linear regression} \label{sec:rescaled_lasso}

The Lasso is a popular technique for regression with implicit feature selection \citep{Tibshirani1996}.
In this section, we demonstrate that rescaling the set of feature vectors $\{\x_1, \ldots, \x_{n+m}\}$ prior to
the train-validation split so that each feature has variance one, may bias the validation error with respect to the generalization error.

\paragraph{Sampling distribution:}
First we generate a random vector of coefficients $\boldsymbol{\beta} = (\beta_1, \ldots, \beta_p)^T$ where $\beta_i \sim \mathcal{N}(0,1)$.
Then we draw each observation $(\x,y)$ in the following manner,
\begin{align}
    \x &\sim \mathcal{N}(0, I_{p \times p}), \label{eq:gaussiandesign}
    \qquad
    y = \x \boldsymbol{\beta} + \epsilon
    \qquad
    \text{where } \epsilon \sim \mathcal{N}(0,\sigma^2).
\end{align}
\paragraph{Preprocessing:}
We estimate the variance of the $j$\textsuperscript{th} coordinate vector\\ $\{x_{1,j}, \ldots, x_{n+m,j}\} \in \mathbb{R}^{n+m}$ as follows,
\begin{align} \label{eq:rl_variance_estimator}
    \hat{\sigma}_j^2 := \frac{1}{n+m} \sum_{i=1}^{n+m} x_{i,j}^2.
\end{align}
Then we rescale the $j$\textsuperscript{th} coordinate of every covariate  by $\hat{\sigma}_j$,
\(
    \hatT(\x) := (x_1 / \hat{\sigma}_1, \ldots, x_p / \hat{\sigma}_p). 
\)
We use the estimate in Eq. \eqref{eq:rl_variance_estimator} because it is easier to analyze mathematically than the standard variance estimate,
but gives similar results in simulations.
\paragraph{Predictor:}
The predictor is $\hat{f}(\hatT(\x)) = \hatT(\x) \hat{\boldsymbol{\beta}}^{\text{Lasso}}$ where the coefficients vector is obtained by Lasso linear regression,
\begin{align}
    \hat{\boldsymbol{\beta}}^{\text{Lasso}}
    &:=
    \argmin_{\boldsymbol{\beta} \in \mathbb{R}^p} \frac{1}{2} \| Y - X \boldsymbol{\beta} \|^2 + \lambda \|\boldsymbol{\beta}\|_1 .
\end{align}
Here, $X_{n \times p}$ is the design matrix with rows comprised of the rescaled training covariates $\hatT(\x_1), \ldots, \hatT(\x_n)$,
the responses vector is $Y = (y_1, \ldots, y_n)^T$ and $\lambda>0$ is a constant that controls the regularization strength.

\subsubsection{Simulation study}

\begin{figure}
    \includegraphics[width=0.5\linewidth]{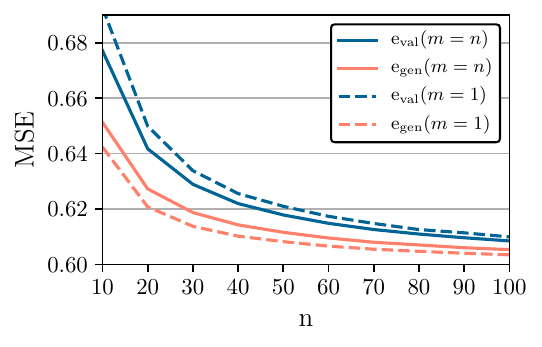}
    \includegraphics[width=0.5\linewidth]{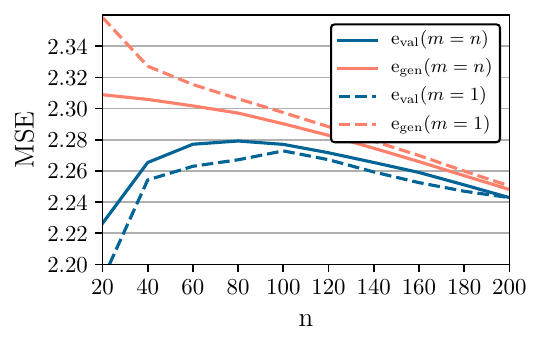}
    \caption{Validation and generalization errors of the rescaled Lasso.
    Solid lines correspond to a validation set of size $m=n$ whereas dashed lines correspond to a singleton validation set.
    (left) Low-dimensional setting, $p=5, \lambda=0.5, \sigma=0.1$, average of 10,000,000 runs;\qquad\qquad\qquad
    (right) High-dimensional setting, $p=10,000, \lambda=0.1, \sigma=1$, average of 1,000,000 runs.
    \label{fig:rescaled_lasso}}
\end{figure}

Figure \ref{fig:rescaled_lasso} shows averaged validation and generalization errors of 
both high-dimensional and low-dimensional Lasso linear regression with preliminary rescaling.
Error bars were omitted as the uncertainty is negligible.
Note  that using $m=1$ validation samples incurs a larger absolute bias than $m=n$ samples.
This observation agrees with our analysis in the next section, since the correlation between $x_{n+1,1}^2$ and $\hat{\sigma}_1$
is stronger when the there is just a single validation data point.

\subsubsection{Analysis}

The Lasso is difficult to analyze theoretically since a closed form expression is not available in the general case.
In order to gain insight, we consider instead a simplified Lasso procedure which performs $p$ separate one-dimensional Lasso regressions,
\begin{align} \label{eq:simplified_lasso}
    \hat{\beta}^{\text{SL}}_j
    &=
    \argmin_{\beta \in \mathbb{R}} \frac{1}{2} \sum_{i=1}^n \left( Y_i - \beta X_{i,j} \right)^2 + \lambda |\beta|
.
\end{align}
The resulting regression coefficients are a soft-threshold applied to the coefficients of simple linear regression.
With this simplification, we are able to analyze the bias.
\begin{theorem} \label{prop:rescaled_lasso_biasp}
    Let clip$_a(z) = \text{max}(\text{min}(z,a),-a)$ denote the truncation of $z$ to the interval $[-a,+a]$.
    Under the simplifying assumptions of zero noise,
    the rescaled simplified Lasso of Eq. \eqref{eq:simplified_lasso} with sampling distribution defined above has the following bias due to the feature rescaling, 
    \begin{align} \label{eq:bias_lasso_orthogonal}
        \bias = p \cdot \E\mathrm{Cov}(\text{clip}_{\lambda \hat{\sigma}_1/n}^2(\beta_1), x_{n+1,1}^2).
    \end{align}
    where $\beta_1$ is the true regression coefficient of the first dimension.
\end{theorem}
The proofs of this theorem and the next one are included in the Appendix.
Large values of $x_{n+1,1}^2$  positively correlate with large values of $\hat{\sigma}_1$,
which positively correlate with $\text{clip}_{\lambda \hat{\sigma}_1/n}^2(\beta_1)$.
We thus expect that this covariance be positive.
This is formally proved in the following theorem.
\begin{theorem} \label{prop:rescaled_lasso_bias_negative}
    Under the assumptions of Theorem \ref{prop:rescaled_lasso_biasp},
    it follows that 
    \(
        \bias > 0
    \).
\end{theorem}

To connect this analysis with our simulation results, consider the left panel of Figure~\ref{fig:rescaled_lasso}.
It shows that in the low-dimensional setting for both $m=n$ and $m=1$ the validation error is uniformly larger than the generalization error, in accordance with Theorem \ref{prop:rescaled_lasso_bias_negative}.
However, this is not true for the high-dimensional case shown in the right panel.
We can explain this by noting that the Lasso applied to an orthogonal design is nothing but a soft-threshold applied to each
linear regression coefficient \citep{Tibshirani1996}, just like the simplified Lasso we analyze here.
Recall that our covariates are Gaussian.
In the low-dimensional case, the $5 \times 5$ matrix $X^T X$ is unlikely to have large deviations from its expected value $n I_{5 \times 5}$, but this is not true for the high-dimensional case.
Thus in the context of our simulations, the simplified Lasso model may serve as a reasonable approximation to the full Lasso but only in the low-dimensional regime.

We note also that in the classic regime where the dimension $p$ is fixed and $n \to \infty$, we obtain $\hat{\sigma}_j \to 1$ for all $j$ and so from
Eq. \eqref{eq:bias_lasso_orthogonal}
we see that $\bias \to 0$.
A similar result was obtained for the feature-selected linear regression analyzed in Section~\ref{sec:feature_selected_linreg} and the categorical grouping example analyzed in Section \ref{sec:categorical}.

\section{Potential impact on  model selection} \label{sec:model_selection}

In many applications the main use of cross-validation is for model selection.
For example, to pick a regularization parameter to use on a given dataset.
This raises the question: can unsupervised preprocessing affect model selection?
In this section, we consider two prototypical model-selection pipelines:
\begin{itemize}
    \item {\bf Correct pipeline:} for each split of the dataset into a training set and a validation set, preprocess the covariates of the dataset using parameters estimated on the training set only.
Then, for each choice of the regularization parameter $\lambda$, train the predictive model on the training folds and compute its loss on the validation fold.

    \item {\bf Incorrect pipeline:} preprocess the entire dataset. For each CV split, train a predictor on the preprocessed training folds and evaluate it on the preprocessed validation fold.
\end{itemize}
In both cases, the regularization parameter $\lambda$ is picked to be the one that minimizes the mean loss across all splits. Then, to obtain the final predictive model, the covariates of the dataset are preprocessed again, this time using the entire data set, and a predictive model is retrained on the preprocessed dataset.

As we have seen in the previous sections of this paper, the incorrect pipeline outlined above can introduce biases into the validation error.
Can this result in suboptimal selection of the regularization parameter $\lambda$?
In general, we can expect that the incorrect preprocessing will alter the optimization landscape of $\lambda$ in subtle ways, thus leading to a different average performance of the resulting models.
To test whether or not this can result in a measurable performance penalty, we compared the correct and incorrect pipelines outlined above on the rescaled lasso example of Section~\ref{sec:rescaled_lasso}. 
For different choices of the training set size $N$, we used 10-fold cross-validation to pick $\lambda \in \{2^{-10}, 2^{-9}, \ldots, 2^{10}\}$ to use in the lasso regression and then estimated the generalization error of the selected model on a holdout set of size $100N$.
The plots shown in Figure~\ref{fig:model_selection} show a small penalty to the average model performance due to incorrect preprocessing.

\begin{figure}
    \includegraphics[width=0.5\linewidth]{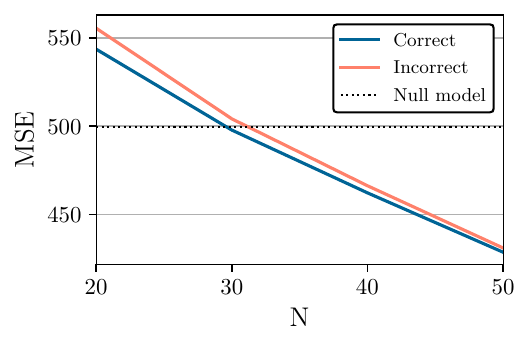}
    \includegraphics[width=0.5\linewidth]{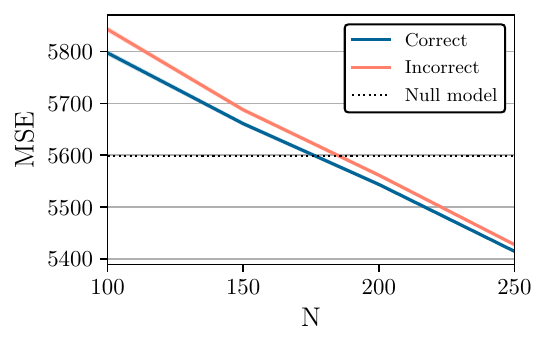}
    \caption{Generalization error of the high-dimensional normalized lasso models produced by the correct and incorrect pipelines described in Section~\ref{sec:model_selection}.
    The data generating model is $y = \x \boldsymbol{\beta} + \mathcal{N}(0, \sigma^2)$ where $\x$ and $\boldsymbol{\beta}$ are drawn from a t-distribution with 4 degrees of freedom. 
    (left) $p=100$, $\sigma=10.0$ with 100,000 repetitions;
    (right) $p=1000$, $\sigma=40.0$, with 50,000 repetitions.}
    \label{fig:model_selection}
\end{figure}

\section{On generic upper bounds of the bias} \label{sec:pathological}
In the examples described in Section \ref{sec:feature_selected_linreg} and \ref{sec:additional_examples},
the bias due to unsupervised preprocessing tends to zero as the sample size tends to infinity, provided that the rest of the model parameters are held fixed.
This suggests the following natural question: can one show that $\bias \xrightarrow{n \to \infty} 0$ for \emph{any} unsupervised preprocessing procedure and any learning algorithm?
In this section, we show that in the most general case the answer is negative.
This is done using a pathological and unnatural preprocessing procedure.
On the bright side, at the end of this section we describe how one might prove that the bias vanishes as $n \to \infty$ for more natural families of preprocessing procedures that admit a uniform consistency property.

Let $\mathcal D$ be a sampling distribution over $\mathcal X \times \mathcal Y$ with a continuous marginal ${\mathcal D}_{\mathcal
X}$.
Recall that in our framework the observations $(\x_1, y_1), \ldots, (\x_{n+m}, y_{n+m})$ are generated i.i.d. from $\mathcal D$
and then an unsupervised transformation $\hatT$ is constructed from $\x_1, \ldots, \x_{n+m}$.
We consider the following unsupervised transformation:
\begin{align} \label{eq:hatT_lowerbound_proof}
    \hatT(\x) =
    \begin{cases}
        \x & \text{if } \x \in \{\x_1, \ldots, \x_{n+m}\}, \\
        \x_0 & \text{otherwise}. \\
    \end{cases}
\end{align}
where $\x_0 \in \mathcal X$ is some point.
By design, when we apply this transformation to the feature vectors in the training and validation sets, they remain intact.
It follows that for any learning algorithm $A_f: (\mathcal X \times \mathcal Y)^n \to (\mathcal X \to \mathcal Y)$,
\begin{align} \label{eq:eval_pathological}
    \E_{S} \,\errval = \E_{\Strain \sim \mathcal D^n}  R( A_f(\Strain))
\end{align}
where $R(f) = \E_{\x,y}[\ell(f(\x), y)]$ is the risk of a prediction rule $f$.
In contrast, for new observations $(\x,y) \sim \mathcal D$, since the marginal distribution $\mathcal D_{\mathcal X}$ is
continuous we have that
\begin{align}
    \pr{\x \in \{\x_1, \ldots, \x_{n+m}\}}=0
\end{align}
and therefore,
    \begin{align} \label{eq:egen_pathological}
        \E \,\errgen
        =
        \E_{\x,y} \ \ell[y, \hat{f} \{ \hatT(\x) \})]
        =
        \E_y \ \ell[y, \hat{f}\{\x_0\}]
        \ge
        \inf_{y'} \ \E_{y} \ell(y, y'). 
    \end{align}
To summarize, the addition of this pathological  unsupervised preprocessing stage has no effect on the expected validation
error, but  the learning procedure can no longer generalize since it has degenerated into a constant predictor.

Equations \eqref{eq:eval_pathological} and \eqref{eq:egen_pathological} allows one to prove lower bounds on the absolute
bias due to unsupervised preprocessing of specific combinations of learning algorithms and sampling distributions.
For a concrete example, consider simple linear regression of noiseless data from the sampling distribution
$x \sim \mathcal{N}(0,1)$ and $y=x$.
Let $\hatT$ be the transformation in \eqref{eq:hatT_lowerbound_proof} with some $x_0$.
For any $n \ge 2$ samples, linear regression will learn the predictor $\hat{f}(x) = x$ and obtain zero loss  on the validation
set,
since for these observations we have $\hat{f}\{\hatT(x)\} = \hat{f}(x) = x = y$.
In contrast, for new observations we have $\hat{f}\{\hatT(x)\}=\hat{f}(x_0) = x_0$. Hence, for the squared loss $\ell(y,y') = (y-y')^2$
we obtain 
\begin{align}
    \E \, \errval = 0,
    \qquad
    \E \, \errgen \ge \inf_{y_0} \E (y-y_0)^2 = 1. 
\end{align}
In this case, choosing $x_0=0$ yields $\E \, \errgen =1$ so the bound is tight.

In light of the above discussion it is clear that, for general unsupervised preprocessing transformations, one cannot attain
upper bounds on $\bias$ that converge to zero as $n \to \infty$.
However, it is certainly possible to do so in more limiting scenarios.
For example,
suppose that the process of learning the unsupervised transformation $\hatT$ is consistent. i.e. there exists
some limiting oracle transformation $T_o$ such that
\begin{align}
    \sup_{\x \in \mathcal X} \|\hatT(\x) - T_o(\x)\| \xrightarrow{n+m \to \infty} 0
\end{align}
in probability.
For example, in the case of a standardizing transformation as discussed in Section \ref{sec:therightway}, the oracle transformation
is $T_o(\x) = (x - \mu)/\sigma$ where $\mu$ and $\sigma$ are the true expectation and standard deviation of the sampling
distribution covariates.
In such cases, one may view the transformed training set $\{(\hatT(\x_1), y_1), \ldots, (\hatT(\x_n), y_n)\}$ as a bounded
perturbation of the limiting set $\{ (T_o(\x_1), y_1), \ldots, (T_o(\x_n), y_n) \}$.
One can then make the connection to the algorithmic robustness literature \citep{XuCaramanisMannor2010,XuMannor2012} or similar
notions and show that a consistent preprocessing transformations followed by a robust learning algorithms guarantees that
$\bias \xrightarrow{n \to \infty} 0$.

\subsection{Upper bounds based on stability arguments}
Another  related approach for upper-bounding the bias due to preprocessing is based on the notion of stability \citep{BousquetElisseeff2002,EvgeniouEtal2004,ShalevshwartzShamirSrebroSridharan2010,HardtRechSinger2016,RussoZou2020}.
In the following we denote a sample by $S = \{(\x_1, y_1), \ldots, (\x_N, y_N)\}$ and by $S^i$  the same sample where $(\x_i, y_i)$ was replaced by $(\x_i', y_i')$.
\begin{definition}
    A learning algorithm $A:(\mathcal X \times \mathcal Y)^N \to (\mathcal X \to \mathcal Y)$ is uniformly stable with parameter $\gamma$ if for any $(\x, y), (\x', y') \in \mathcal X \times \mathcal Y$ and $S \in (\mathcal X \times \mathcal Y)^N$ 
\[
    |\ell(A_S(\x), y) - \ell(A_{S^i}(\x), y)| \le \gamma.
\]
\end{definition}
It has been shown in a series of papers that for a bounded loss function, any uniformly stable learning algorithm gives a predictor with a risk close to the average loss on the training set \citep{BousquetElisseeff2002,FeldmanVondrak2018,FeldmanVondrak2019,BousquetKlochkovZhivotovskiy2020}.

\begin{theorem}[Corollary 7 of \cite{BousquetKlochkovZhivotovskiy2020}] \label{thm:stabilityriskbound}
    For a uniformly stable learning algorithm $A$ with parameter $\gamma$ and a bounded loss function $\ell \le L$,
    for any $\delta \in (0,1)$ the following holds with probability $\ge 1-\delta$,
    \begin{align}
        |R(A_S) - R_{emp}(A_S)| \lesssim \gamma \log n \log \tfrac{1}{\delta} + \frac{L}{\sqrt{n}} \sqrt{ \log \tfrac{1}{\delta} }.
    \end{align}
\end{theorem}
Here, $R(A_S)$ is the risk of the predictor learned from the set $S$ and $R_{emp}(A_S)$ is the mean loss of the same predictor on the training set.
The setting that we study in this paper is slightly different in two regards:
\begin{enumerate}
    \item What we consider to be the learning procedure $A_S$ is in fact a two-step procedure that includes a preprocessing function $\hat T$ learned from the entire data set via $A_T$, followed by the learning procedure $A_f$ that is trained on the preprocessed training set $(\hat{T}(\x_1), y_1), \ldots, (\hat{T}(\x_n), y_n)$. The uniform stability condition must therefore apply to the combined two-step procedure.
    \item We are not interested in bounding the risk with respect to the loss on the training set $\left\{ (\x_1, y_1), \ldots (\x_n, y_n) \right\}$, but rather with respect to the loss on the validation set $\left\{ (\x_{n+1}, y_{n+1}), \ldots, (\x_{n+m}, y_{n+m}) \right\}$.
\end{enumerate}
Nonetheless, the proof of Theorem \ref{thm:stabilityriskbound} can be adapted to our setting in a straightforward manner, which gives the following result,
\begin{theorem}
    Let $A$ be a combined preprocessing+learning algorithm that is uniformly stable with parameter $\gamma$.
    For any $\delta \in (0,1)$ the following  holds with probability $\ge 1-\delta$,
    \begin{align}
        |\errgen - \errval| \lesssim \gamma \log ({n+m}) \log \tfrac{1}{\delta} + \frac{L}{\sqrt{m}} \sqrt{ \log \tfrac{1}{\delta}}.
    \end{align}
    where $n$ is the size of the training set, $m$ is the size of the validation set and $L$ is an upper bound on the loss function.
\end{theorem}
It follows that any combination of a preprocessing and a learning procedure that is uniformly stable is guaranteed to have a small bias due to preprocessing.

\section{Conclusion}
Preliminary data-dependent transformations of data sets prior to predictive modeling and evaluation can introduce a bias
to cross-validation estimators, even when those transformations only depend on the feature vectors in the data set and not on the response.
This bias can result in over-optimistic or under-optimistic estimates of model performance and may lead to sub-optimal model selection.
The magnitude of the bias depends on the specific data distribution, modeling procedure, and cross-validation parameters.
In this paper, we study three examples that involve commonly-used preprocessing transformations
and regression methods.
We summarize some of our key findings:
\begin{enumerate}
    \item The bias due to unsupervised preprocessing may be  positive or negative.
    This depends on the particular parameters of the problem in a non-obvious
manner. For example, in Section \ref{sec:categorical} we demonstrated that the sign of the bias is flipped by changing
the size of the validation set or even by adjusting the level of noise.
    \item When cross-validation is used for model selection, such as picking the best regularization parameter, the presence of unsupervised preprocessing can hurt the average performance of the selected model.
    \item It seems to be typical for the bias to vanish as the number of samples tends to infinity. However, there are counterexamples. In Section \ref{sec:pathological} we gave a pathological example where the bias remains constant
and in Section \ref{sec:feature_selected_linreg} we showed a more realistic example where the bias tends to infinity in a regime where the number of samples is fixed but the dimension tends to infinity.
    \item While the size of the validation set affects the magnitude of the bias, one cannot make the bias vanish by simply manipulating the size of the validation set. In fact, in all of our examples, leave-one-out and two-fold cross-validation showed
roughly similar magnitudes of the bias despite having vastly different validation set sizes.
\end{enumerate}
We believe that in light of these results, the scientific community should re-examine
the use of preliminary data-dependent transformations, particularly when dealing with small sample sizes and high-dimensional data sets.
By default, the various preprocessing stages should be incorporated into the cross-validation scheme as described in Section~\ref{sec:therightway}, thus eliminating any potential biases.
Further research is needed to understand the full impact of preliminary preprocessing in various application domains.

\section*{Reproducibility}
Code for running the simulations and generating exact copies of all the figures in this paper is available at:
\url{https://github.com/mosco/unsupervised-preprocessing}


\section*{Acknowledgments}

We would like to thank Ariel Goldstein, Shay Moran, Kenneth Norman,  Robert Tibshirani and the anonymous reviewers for their invaluable input.
This research was partially supported by Israeli Science Foundation grant ISF 1804/16.


\appendix

\section*{Appendix. Technical proofs}

\section*{Proof of Theorem \ref{prop:var_selected_linreg}}

Recall that the columns are all interchangeable and assume w.l.o.g. that $\widehat{j}=1$. i.e. the first variable was selected.
In that case the preprocessed design matrix is $\tilde{X} = X_{*1}$.
The estimated regression coefficient is
\begin{align} \label{eq:beta_1_est}
    \hat{\beta}_1 = (\tilde{X}^T \tilde X)^{-1} \tilde{X}^T Y = \frac{X_{*1}^T Y}{\| X_{*1}\|^2} 
\end{align}
Recall that $Y = X \boldsymbol{\beta} = \sum_{j=1}^p \beta_j X_{*j}$.
Plugging this into \eqref{eq:beta_1_est} gives
\begin{align}
    \hat{\beta}_1 = \frac{X_{*1}^T Y}{\| X_{*1} \|_2^2}  = \frac{X_{*1}^T \sum_{j=1}^p \beta_j X_{*j}}{\| X_{*1} \|_2^2}
    =
    \beta_1 + \sum_{j=2}^p \beta_j\cdot Z_j
\end{align}
where $Z_j = X_{*1}^T X_{*j} / \|X_{*1} \|_2^2$.
It follows that the prediction for the first observation in the validation set is,
\begin{align}
    \hat{y}_{n+1} &= \hat{\beta}_1 X_{n+1,1} = \left( \beta_1 + \sum_{j=2}^p \beta_j Z_j\right) X_{n+1,1}.
\end{align}
Our model is noiseless and satisfies $y = \x \boldsymbol{\beta}$, hence
\(
    y_{n+1} = \sum_{j=1}^p \beta_j X_{n+1,j}.
\)
The mean squared error of the single validation sample is thus
\begin{align}
    MSE &= \E(\hat{y}_{n+1} - y_{n+1})^2 \\
    &= \E \left\{ \sum_{j=2}^p \beta_j \left( Z_j X_{n+1,1} - X_{n+1,j} \right) \right\}^2 \\
    &= \sum_{j=2}^p \sum_{\ell=2}^p \E \left\{ \beta_j \beta_\ell \left( Z_j X_{n+1,1} - X_{n+1,j} \right) \left( Z_\ell X_{n+1,1}
- X_{n+1,\ell} \right) \right\} \\
    &= \sum_{j=2}^p \sum_{\ell=2}^p \E \left( \beta_j \beta_k \right) \E \left\{ \left( Z_j X_{n+1,1} - X_{n+1,j} \right)
\left( Z_\ell X_{n+1,1}
- X_{n+1,\ell} \right) \right\}
\end{align}
where the last inequality follows from the independence of $\beta_j$. Furthermore, since \\$\beta_1, \ldots, \beta_p\sim\mathcal{N}(0,1)$
it follows that $\E \beta_j \beta_\ell = \delta_{j\ell}$. The MSE simplifies to
\begin{align} \label{eq:blah}
    \sum_{j=2}^p  \E  \left( Z_j X_{n+1,1} - X_{n+1,j} \right)^2 = \sum_{j=2}^p \left\{ \E Z_j^2 X_{n+1,1}^2  -2 \E
\left( Z_j X_{n+1,1} X_{n+1,j}\right) + \E X_{n+1,j}^2 \right\}.
\end{align}
We now claim that $\E  Z_j X_{n+1,1} X_{n+1,j} =0$ by a symmetry argument.
To see why, let $X_{1 \ldots n+m}$ be the random matrix of the training and validation covariates and let $X'_{1 \ldots n+m}$ be the same matrix with the sign of $X_{n+1,1}$ flipped.
Since we assumed that the covariates are drawn i.i.d. $\mathcal{N}(0,1)$, the distribution of $X_{1 \ldots n+m}$ is equal to that of $X'_{1 \ldots n+m}$.
Note that $Z_j$ depends only on the first $n$ observations, so it is not affected by the sign flip.
Therefore, we must have \( \E Z_j X_{n+1,1} X_{n+1,j} = \E Z_j (-X_{n+1,1}) X_{n+1,j} \). A subtle point in this argument is that the variable selection preprocessing step must not be affected by the sign flip of $X_{n+1,1}$. This is the reason we opted to analyze variable selection based on the sum of squares rather than the variance.
We rewrite the last term of \eqref{eq:blah} as,
\begin{align}
    \sum_{j=2}^p \E X_{n+1,j}^2 = \sum_{j=1}^p \E X_{n+1,j}^2 - \E X_{n+1,1}^2 = p  - \E X_{n+1,1}^2.
\end{align}
As for the first term on the right hand side of \eqref{eq:blah}, since columns $2, \ldots, p$ are identically distributed, we have \( \sum_{j=2}^p \E Z_j^2 X_{n+1,1}^2  = (p-1) \E  Z_2^2 X_{n+1,1}^2  \).
Recall that $\hat{\rho}_{1,2}$ is a random variable equal to the cosine of the angle between the first and second
features in the training set covariates. This variable is independent of the column magnitudes, hence it is independent of $X_{n+1,1}^2$ even after
conditioning
on the selection $\widehat{j}=1$. Hence,
\begin{align}
    &\E Z_2^2 X_{n+1,1}^2 
    =
    \E \hat{\rho}_{1,2}^2 \frac{\|X_{*2}\|^2}{\|X_{*1}\|^2} X_{n+1,1}^2 
\end{align}
To summarize, we have
\begin{align}
    &\E \errval = \E(\hat{y}_{n+1} - y_{n+1})^2 
    = (p-1)  \E  \hat{\rho}_{1,2}^2\frac{\|X_{*2}\|^2}{\|X_{*1}\|^2} X_{n+1,1}^2  + p  - \E X_{n+1,1}^2.
\end{align}
Similarly, for a new holdout observation $(\x, y)$ we have
\begin{align}
    \E \errgen
    &=
    \E(\hat{y} - y)^2
    =
    (p-1) \E \hat{\rho}_{1,2}^2 \frac{\|X_{*2}\|^2}{\|X_{*1}\|^2} \cdot \E x^2 + p  - 1.
\end{align}
Hence,
\(
    \E \errval - \E \errgen = (p-1) \E  \hat{\rho}_{1,2}^2\frac{\|X_{*2}\|^2}{\|X_{*1}\|^2} \left( X_{n+1,1}^2 - x^2\right)
    -
    \E X_{n+1,1}^2 + 1
\).
where it was assumed that $\widehat{j} = 1$ w.l.o.g.
The result follows.
$\hspace*{\fill} \square$

\section*{Proof of Corollary \ref{cor:var_selected_linreg}}
Again, we assume w.l.o.g. that $\widehat{j} = 1$ and arbitrarily set $j_o = 2$.
We denote convergence in probability by $\convprob$ and convergence in distribution by $\convdist$.
\paragraph{($p$ fixed, $n \to \infty$):}

In this case $\|X_{*2}\|/\|X_{*1}\| \convprob 1$ and $\hat{\rho}_{1,2} \convprob 0$,
thus as $n \to \infty$ 
\begin{align}
    (p-1) \hat{\rho}_{1,2}^2 \frac{\|X_{*2}\|^2}{\|X_{*1}\|^2} \convprob 0.
\end{align}
To analyze the terms that involve $X_{n+1,1}^2$, recall that we assumed w.l.o.g. that the first column of X is the one with
the highest variance.
Let us denote its norm by
\begin{align}
    R_N = \|(X_{1,1}, \ldots, X_{N,1})\| \text{ where } N=n+m
\end{align}
If we condition on $R_N,$ the vector $(X_{1,1}, \ldots, X_{N,1})$ is sampled uniformly from the $n-1$-sphere of radius $R_n$.
Hence $X_{n+1,1} = R_N \cdot S_N$ where $S_N$ is the first coordinate of a random point on the unit $n-1$-sphere.
    It is well-known that
\begin{align} \label{eq:nsphere_first_coordinate}
    \tfrac12 (S_N+1) \sim \text{Beta}\left(N/2, N/2\right).
\end{align}
With proper normalization, the Beta distribution, converges to a Gaussian as $N \to \infty$
(see for example Lemma A.1 of \cite{MoscovichNadlerSpiegelman2016}). In our case,
\begin{align}
    \sqrt{N}\text{Beta} \left( N/2, N/2 \right)
    \convdist
    \mathcal{N} \left( 1/2, 1/4 \right).
\end{align}
It follows that $\sqrt{N}S_N \convdist \mathcal{N}(0,1)$ and since $R_N/\sqrt{N} \convprob 1$, by Slutsky's theorem we see
that $X_{n+1,1} = R_N S_N \convdist \mathcal{N}(0,1)$ and therefore $X_{n+1,1}^2 \convdist \chi^2(1)$. Applying Slutsky's
theorem again, we obtain
\begin{align}
    \left\{
        (p-1) \hat{\rho}_{1,2}^2 \frac{\|X_{*2}\|^2}{\|X_{*1}\|^2} - 1
    \right\}
    \left(
        X_{n+1,1}^2 - \E x^2 
    \right)
    \convdist
    1-\chi^2(1)
\end{align}
and therefore $\bias \to 0$.
\paragraph{($n,m$ fixed, $p \to \infty$):}
By Eq. \eqref{eq:bias_var_selected_linreg},
\begin{align} \label{eq:bias_decomposed}
        \bias
        =
        \E \left\{ (p-1) \hat{\rho}_{1,2}^2 \frac{\|X_{*2}\|^2}{\|X_{*1}\|^2}X_{n+1,1}^2 \right\}
        -
        \E \left\{ (p-1) \hat{\rho}_{1,2}^2 \frac{\|X_{*2}\|^2}{\|X_{*1}\|^2} \right\}
        +
        \E X_{n+1,1}^2
        +
        1
\end{align}
We start with the first term. It is easy to see that as $p \to \infty$ we have,
\begin{align}
    R_N \to \infty,
    \qquad
    \frac{\|X_{*2}\|^2}{n} \convprob 1,
    \qquad
    \frac{\|X_{*1}\|^2}{R_N^2} \convprob \frac{n}{N}.
\end{align}
Since $X_{n+1,1} = R_N S_N$, it follows that
\begin{align}
    \frac{\|X_{*2}\|^2}{\|X_{*1}\|^2} X_{n+1,1}^2
    =
    \frac{N}{R_N^2} (R_N S_N)^2\{1+o_P(1)\}
    \convdist \chi^2(1)
\end{align}
Note that $\hat{\rho}_{1,2}^2$ is independent of $\frac{\|X_{*2}\|^2}{\|X_{*1}\|^2} X_{n+1,1}^2$ and therefore,  
\begin{align}
    &\E \left\{ (p-1) \hat{\rho}_{1,2}^2 \frac{\|X_{*2}\|^2}{\|X_{*1}\|^2} X_{n+1,1}^2 \right\} \\
    &=
    (p-1)\E \hat{\rho}_{1,2}^2 \E \frac{\|X_{*2}\|^2}{\|X_{*1}\|^2} X_{n+1,1}^2
    \xrightarrow{p \to \infty}
    (p-1) \E \hat{\rho}_{1,2}^2 \E \chi^2(1) = (p-1) \E \hat{\rho}_{1,2}^2.
\end{align}
Recall that $\hat{\rho}_{1,2}$ is nothing but a normalized dot product of two independent Gaussians.
Since dot products are invariant to rotations of both vectors, we may assume w.l.o.g. that the first vector lies on the first
axis and conclude that $\hat{\rho}_{1,2}$ is distributed like the first coordinate of a uniformly drawn unit vector on the $n-1$-sphere.
By Eq. \eqref{eq:nsphere_first_coordinate} we see that $\E[\hat{\rho}_{1,2}] = 0$ and
\begin{align}
    \E[\hat{\rho}_{1,2}^2]
    =
    \var (\hat{\rho}_{1,2})
    =
    \var \{2 \text{Beta}(\tfrac{n-1}{2}, \tfrac{n-1}{2})\} = \tfrac{1}{n}.
\end{align}
It follows that
\begin{align}
    \E \left\{ (p-1) \hat{\rho}_{1,2}^2 \frac{\|X_{*2}\|^2}{\|X_{*1}\|^2} X_{n+1,1}^2 \right\}
    =
    \frac{p-1}{n}.
\end{align}
The second term of Eq. \eqref{eq:bias_decomposed} is negligible with respect to this, since 
\begin{align}
    \E \left\{ (p-1) \hat{\rho}_{1,2}^2 \frac{\|X_{*2}\|^2}{\|X_{*1}\|^2} \right\} = \frac{p-1}{n} \E \frac{\|X_{*2}\|^2}{\|X_{*1}\|^2}
\end{align}
and in the regime that $n$ is fixed and $p \to \infty$ we have $\E \left[ \|X_{*2}\|^2 / \|X_{*1}\|^2\right] \convprob 0$.
As for the third term of Eq. \eqref{eq:bias_decomposed}, note that
\begin{align}
    \E X_{n+1,1}^2 < \E \max_{i=1, \ldots, N,\ j=1,\ldots,p} X_{i,j}^2 .
\end{align}
It is well known that the expectation of the maximum of $k$ chi-squared variables is $O(\log k)$, hence
\(
    \E X_{n+1,1}^2 =  O(\log N + \log p)
\).
Thus we conclude that in the regime of fixed $n,m$ and $p \to \infty$,
\(
    \bias = \frac{p}{n} \{1 + o(1) \} \to \infty.
\)
$\hspace*{\fill} \square$
\section*{Proof of Theorem \ref{prop:rare_categories_MSE}}
We analyze 3 distinct cases:
\paragraph{Case 1:} the category $k$ is rare.
Hence its predicted response is $\hat{f}(k)=0$.
Since the mean response of category $k$ is $\mu_k$, the MSE for predicting the response of a sample
with response $\mu_k + \mathcal{N}(0,\sigma^2)$ is $\sigma^2 + \mu_k^2$.
\paragraph{Case 2:} $k$ is not rare but $\#_n(k)=0$.
Again, the predicted response is zero, leading to the same MSE as in Case 1.
\paragraph{Case 3:} $k$ is not rare and $\#_n(k) \ge 1$.
The predicted response will be the mean of $\#_n(k)$ responses from the training set.
The distribution of this mean is $\mathcal{N}(\mu_k, \sigma^2/\#_n(k))$,
Hence the expected MSE is $\sigma^2 + \sigma^2 / \#_n(k)$.

Combining these 3 cases, we obtain
\begin{align}
    \pr{r_k} \left( \sigma^2 + \mu_k^2 \right)
    +
    \pr{\neg r_k} p_0(k)\left( \sigma^2 + \mu_k^2 \right)
    +
    \pr{\neg r_k} \sum_{i=1}^n p_i(k) \left( \sigma^2 + \tfrac{\sigma^2}{i}\right)
    \\
    =
    \sigma^2
    +
    \pr{r_k} \mu_k^2 + \pr{\neg r_k} p_0(k) \mu_k^2
    +
    \sigma^2 \pr{\neg r_k} \sum_{i=1}^{n}  \frac{p_i(k)}{i}.
\end{align}
$\hspace*{\fill} \square$

\section*{Proof of Theorem \ref{prop:rescaled_lasso_biasp}}

First, we define the shrinkage operator, also known as a soft-thresholding operator. For any $x \in \mathbb{R}$ and $a \ge 0$ it shrinks the
absolute value of $x$ by $a$,
or if $|x| \le a$ it returns zero.
\begin{align}
    \text{shrink}_a(x) := sign(x)(|x|-a)^+ \quad \text{ where } \quad (x)^+ := \text{max}(0,x)
\end{align}
We also define a clipping operator, which clips $x$ to the interval $[-a, a]$,
\begin{align}
    \text{clip}_a(x)  := \max(\min(x, a), -a)).
\end{align}
The following identities are easy to verify. For any $x \in \mathbb{R}$ and any $a,c>0$,
\begin{align}
    \text{shrink}_{ca}(cx) = c \cdot \text{shrink}_a(x) \label{eq:shrink_property} \\
    \text{clip}_a(x) + \text{shrink}_a(x) = x. \label{eq:clip_property}
\end{align}
Let $\tilde{X}_{n,p}$ denote  the rescaled design matrix, and let $\Sigma = \text{diag}(\hat{\sigma}_1, \ldots, \hat{\sigma}_p)$.
The preprocessing stage rescales the j\textsuperscript{th} column of $X$ by $1/\hat{\sigma}_j$, hence
\(
    \tilde{X} = X \Sigma^{-1}.
\)
Consider the least-squares solution
\begin{align}
    \hat{\beta}^{\text{OLS}} = \argmin_{\beta \in \mathbb{R}^p} \frac{1}{2} \| Y - \tilde{X} \beta \|^2.
\end{align}
For the rescaled orthogonal design, the solution is $ \hat{\beta}^{\text{OLS}} = \frac{1}{n}\Sigma X^T Y $ and since we are in the noiseless setting $Y = X \beta$, so we have $ \hat{\beta}^{\text{OLS}} = \frac{1}{n} \Sigma X^T X \beta = \Sigma \beta$.
The simplified Lasso procedure, which applies to Lasso to each dimension separately, yields the  solution,
\begin{align}
    \hat{\beta}_j^{\text{Lasso}} = \text{shrink}_{\lambda \hat{\sigma}_j^2/n}(\hat{\beta}^{\text{OLS}}_j) = \text{shrink}_{\lambda \hat{\sigma}_j^2/n}( \hat{\sigma}_j \beta_j) .
\end{align}
We may rewrite this using the property of the shrinkage operator in \eqref{eq:shrink_property},
\begin{align} \label{eq:beta_j_lasso}
    \hat{\beta}_j^{\text{Lasso}} = \hat{\sigma}_j \text{shrink}_{\lambda \hat{\sigma}_j/n}( \beta_j).
\end{align}
Consider the generalization error conditioned on a draw of $\beta$ and the estimated variances,
\begin{align}
    \errgen|\beta, \hat{\sigma}
    &=
    \E_{\x,y}[y - \hat{f}\{\hatT (\x)\}]^2 \\
    &=
    \E
    \left[
        \x^T \beta -  \x^T \Sigma^{-1} \hat{\beta}^{\text{Lasso}}
        |
        \beta, \hat{\sigma}
    \right]^2  \\
    &=
    \E
    \left[
        \sum_{j=1}^p
        \left(
            \beta_j - \text{shrink}_{\lambda \hat{\sigma}_j/n}( \beta_j) 
        \right)
        x_j
        \Big|
        \beta, \hat{\sigma}
    \right]^2 && \text{(By \eqref{eq:beta_j_lasso})} \\
    &=
    \E
    \left[
        \sum_{j=1}^p \text{clip}_{\lambda \hat{\sigma}_j/n}(\beta_j) x_j
        \Big|
        \beta, \hat{\sigma}
    \right]^2  && \text{(By \eqref{eq:clip_property})}\\
    &=
    \sum_{j=1}^p
    \sum_{k=1}^p
    \text{clip}_{\lambda \hat{\sigma}_j/n}(\beta_j) \text{clip}_{\lambda \hat{\sigma}_k/n}(\beta_k)
    \E
    \left[
        x_j x_k
    \right]. \label{eq:lasso_egen}
\end{align}
Recall that $x_j \simiid \mathcal{N}(0,1)$, hence $\E x_j x_k = \delta_{j,k}$. Thus,
\(
    \errgen|\beta, \hat{\sigma}
    =
    \sum_{j=1}^p
    \text{clip}_{\lambda  \hat{\sigma_j}/n}^2(\beta_j).
\)
To compute the expected generalization error, one must integrate this with respect to the probability density of $\beta$
and $\hat{\sigma}$. Since all coordinates are identically distributed, it suffices to integrate with respect to the first coordinate,
\(  \E \errgen
    =
    p \cdot \E_{\beta_1, \hat{\sigma}_1}
    \text{clip}_{\lambda \hat{\sigma}_1/n}^2(\beta_1)
\).
For the validation error,
\( \errval|\beta, \hat{\sigma}
    =
    \sum_{j=1}^p
    \sum_{k=1}^p
    \expect{\text{clip}_{\lambda \hat{\sigma}_j/n}(\beta_j)
    \text{clip}_{\lambda \hat{\sigma}_k/n}(\beta_k)
    x_{n+1,j}
    x_{n+1,k}
    | \beta, \hat{\sigma}
    } 
\).
Under our assumptions, the training set covariates $\x_1, \ldots, \x_n$ satisfy an orthogonal design,
however the validation samples $\x_{n+1}, \ldots, \x_{n+m}$ are independent gaussians.
Let $(\errval|\beta, \hat{\sigma})_{j,k}$ denote the $j,k$ term of the double sum above.
For any $j\neq k$,
\begin{align}
    (\errval|\beta, \hat{\sigma})_{j,k}
    &=\expect{
        \text{clip}_{\lambda \hat{\sigma}_j/n}(\beta_j)
        \text{clip}_{\lambda \hat{\sigma}_k/n}(\beta_k)
        x_{n+1,j}
        x_{n+1,k} | \beta, \hat{\sigma}
    }\\
    &=
    \text{clip}_{\lambda \hat{\sigma}_j/n}(\beta_j)
    \text{clip}_{\lambda \hat{\sigma}_k/n}(\beta_k)
    \expect{x_{n+1,j} x_{n+1,k} | \hat{\sigma}}\\
    &=
    \text{clip}_{\lambda \hat{\sigma}_j/n}(\beta_j)
    \text{clip}_{\lambda \hat{\sigma}_k/n}(\beta_k)
    \expect{x_{n+1,j} | \hat{\sigma}_j}
    \expect{x_{n+1,k} | \hat{\sigma}_k}.
\end{align}
The second equality follows from the fact that $\text{clip}_{\lambda \hat{\sigma}_j/n}(\beta_j)$ is constant, conditioned on $\beta, \hat{\sigma}$.
Due to symmetry, it must be the case that $\expect{x_{n+1,j} | \hat{\sigma}_j}$ = $\expect{-x_{n+1,j} | \hat{\sigma}_j} = 0$.
Therefore the above expectation is zero for any $j \neq k$.
However, for $j=k$ we have
\begin{align}
    (\errval|\beta, \hat{\sigma})_{j=k}
    =
    \expect{
        \text{clip}^2_{\lambda \hat{\sigma}_j/n}(\beta_j)
        x^2_{n+1,j}
        |
        \beta_j, \hat{\sigma}_j
    }.
\end{align}
Since all coordinates are equally distributed, we express the expected validation error in terms of the first coordinate.
\begin{align}
    \E \errval
    &=
    p \cdot \E_{\beta_1, \hat{\sigma}_1, x_{n+1,1}} \{ \text{clip}_{\lambda \hat{\sigma}_1/n}^2(\beta_1) x_{n+1,1}^2 \}
    \\
    &=
    p \cdot \E  \text{clip}_{\lambda \hat{\sigma}_1/n}^2(\beta_1)
    \E x_{n+1,1}^2
    +
    p \cdot \text{Cov}(\text{clip}_{\lambda \hat{\sigma}_1/n}^2(\beta_1), x_{n+1,1}^2) \\
    &=
    \E \errgen + p \cdot \E\text{Cov}(\text{clip}_{\lambda \hat{\sigma}_1/n}^2(\beta_1), x_{n+1,1}^2).
\end{align}
$\hspace*{\fill} \square$
\section*{Proof of Theorem \ref{prop:rescaled_lasso_bias_negative}}
We begin with a technical lemma that is concerned with the covariance of two random variables that have a monotone dependency.
\begin{lemma} \label{lemma:monotone_decreasing}
    Let $Y,Z$ be random variables such that $\E(Y|Z)$ is monotone-increasing, then
    \(
        \mathrm{Cov}(Y,Z) > 0.
    \)
\end{lemma}
\begin{proof}
We rewrite $\E(YZ)$ using the law of iterated expectation,
\begin{align}
    \E(YZ) = \E_Z \E(YZ|Z) = \E_Z ( Z \E(Y|Z)).
\end{align}
Both $Z$ and $\E(Y|Z)$ are increasing functions of $Z$. 
By the continuous variant of Chebyshev's sum inequality,
\(
    \E\{Z\E(Y|Z)\} > \E Z \cdot \E\{\E(Y|Z)\}.
\)
and by the law of iterated expectation, the right-hand side is equal to $\E Z \cdot \E Y$.
\end{proof}

Now, let the random variables $Z,Y$  denote $x_{n+1,1}^2$ and $\text{clip}_{\lambda \hat{\sigma}_1/n}^2(\beta_1)$ respectively.
To prove the theorem, we need to show that the following inequality holds.
\begin{align} \label{eq:cov_negative}
    \text{Cov}(\text{clip}_{\lambda \hat{\sigma}_1/n}^2(\beta_1),x_{n+1,1}^2) = \E(ZY)  - \E Z \cdot \E Y > 0.
\end{align}
We will  show that $\E(Y|Z)$ is a monotone-increasing function of Z. \eqref{eq:cov_negative} will then follow from Lemma \ref{lemma:monotone_decreasing}.
For every $Z$ the conditioned random variable $Y|Z$ is non-negative. We may rewrite its expectation using the integral of the tail probabilities,
\begin{align}
    \E(Y|Z) = \int_0^{+\infty} \Pr(Y \ge t|Z)dt.
\end{align}
From the definition of the clipping function, we have
\begin{align}
    \Pr(Y \ge t|Z)
    =
    \Pr(
        \lambda^2 \hat{\sigma}_1^2/n^2 \ge t \text{ and } \beta_1^2 \ge t | Z
    ).
\end{align}
The coefficient $\beta_1$ is drawn independently of the covariates, hence is independent of $Z$ and also independent of $\hat{\sigma}_1$.
Therefore the probability of the conjunction is the product of probabilities,
\(
    \Pr(
        \lambda^2 \hat{\sigma}_1^2 /n^2 \ge t \text{ and } \beta_1^2 \ge t | Z
    )
    =
    \Pr(
        \lambda^2 \hat{\sigma}_1^2 / n^2 \ge t |Z
    )
    \cdot
    \Pr(
        \beta_1^2 \ge t
    ).
\)
Putting it all together, we have shown that
\begin{align}
    \E(Y|Z) = \int_0^\infty \Pr(\lambda^2 \hat{\sigma}_1^2 /n^2\ge t | Z) \cdot \Pr(\beta_1^2 \ge t) dt.
\end{align}
To prove that $\E(Y|Z)$ is monotone-increasing as a function of $Z$,
it suffices  to show that $\Pr(\lambda^2 \hat{\sigma}_1^2/n^2 \ge t | Z) = \Pr(\hat{\sigma}_1^2 \ge n^2 t/\lambda^2 | x_{n+1,1}^2)$
is a monotone-increasing function of $x_{n+1,1}^2$.
Recall that $\hat{\sigma_1}^2 = \tfrac{1}{n+m}\sum_{i=1}^{n+m} x_{i,1}^2 $.
We have
\begin{align}
    \Pr(\hat{\sigma_1}^2 \ge t | x_{n+1,1}^2 = s)
    =
    \Pr(x_{1,1}^2 +\ \ldots + x_{n,1}^2 + 0 + x_{n+2,1}^2 + x_{n+m,1}^2 \ge t-s).
\end{align}
Since all of these variables are independent Gaussians, the probability is monotone-increasing in $s$. This completes the proof.  $\hspace*{\fill} \square$

\section*{Supplementary. Prevalence of unsupervised preprocessing in science}

\begin{figure}
\caption{\label{tbl:science_articles} Science Magazine articles that match the search term ``cross-validation" published
between January 1\textsuperscript{st} 2017 and July 1\textsuperscript{st} 2018.}
\centering
\fbox{%
\begin{tabular}{lcc} 
Article& Performs cross-validation? & Preprocessing? \\
\hline
\hline
\cite{Davoli2017} & Yes & Yes \\\hline
\cite{Cederman2017} & No & \\\hline
\cite{Kennedy2017} & Yes& No\\\hline
\cite{Keller2017} & Yes & No \\\hline
\cite{Liu2017} & Yes & Yes\\\hline
\cite{Caldieri2017} & No & \\\hline
\cite{Leffler2017} & Yes & No \\\hline
\cite{Zhou2017} & No & \\\hline
\cite{Chatterjee2017} & No& \\\hline
\cite{Klaeger2017} & Yes & No \\\hline
\cite{George2017} & Yes & No \\\hline
\cite{Lerner2018} & No & \\\hline
\cite{Delgado-Baquerizo2018} & No & \\\hline
\cite{Ni2018} & Yes & Yes\\\hline
\cite{Dakin2018} & Yes & Yes \\\hline
\cite{Nosil2018} & Yes & No\\\hline
\cite{Cohen2018} & Yes & Yes \\\hline
\cite{Athalye2018} & Yes & No\\\hline
\cite{Mills2018} & Yes & No \\\hline
\cite{Dai2018} & No &  \\\hline
\cite{Mi2018} & Yes & No \\\hline
\cite{Ahneman2018} & Yes & Yes \\\hline
\cite{VanMeter2018} & No & \\\hline
\cite{Barneche2018} & Yes & No\\\hline
\cite{Poore2018} & Yes & No\\\hline
\cite{Scheib2018} & Yes & Yes\\\hline
\cite{Ngo2018} & Yes & No\\\hline
\cite{Ota2018} & Yes & No\\\hline\hline
{Total Yes count} &  20 & 7
\end{tabular}}
\end{figure}

To estimate the prevalence of unsupervised preprocessing prior to cross-validation,
we examined all of the papers published in \textit{Science} between January 1\textsuperscript{st} 2017 and July 1\textsuperscript{st}
2018
which reported performing cross-validation.
To obtain the list of articles we used the advanced search page \url{http://science.sciencemag.org/search}
with the  search term ``cross-validation'' in the above-mentioned time period, limiting the search to \textit{Science Magazine}.
This resulted in a list of 28 publications. However, only 20 of those actually analyze data using a cross-validation (or
train-test)  procedure.
We  read these 20 papers and discovered that at least 7 of them seem to be doing some kind of unsupervised preprocessing
prior to cross-validation.
Hence, 35\% of our sample of papers may suffer from a bias in the validation error!
This  result is based on our  understanding of the data processing pipeline in papers from diverse fields, in most of which
we are far from experts.
Hence our observations may contain errors (in either direction).
The main takeaway is that unsupervised preprocessing prior to cross-validation is very common in high impact scientific publications.
See Table \ref{tbl:science_articles} for the full list of articles that we examined.
In the rest of this section, we describe the details of the unsupervised preprocessing stage in the 7 articles that we believe
perform it.

\paragraph{Tumor aneuploidy correlates with markers of immune evasion and with reduced response to immunotherapy} \cite{Davoli2017}
This work involves a large number of statistical analyses. In one of them, a continuous feature is transformed
into a discrete feature, using percentiles calculated from the entire data set. The transformed feature is then
incorporated into a Lasso model.
This is described in the ``Lasso classification method" section of the article:

``We defined the tumors as having low or high cell cycle or immune signature scores using the 30th and 70th percentiles,
as described above, and used a binomial model.
[...]
We divided the data set into a training and test set, representing two thirds and the remaining one third of the data set,
respectively, for each tumor type. We applied lasso to the training set using 10-fold cross validation."

\paragraph{CRISPRi-based genome-scale identification of functional long noncoding RNA loci in human cells \cite{Liu2017}}

In this paper, the data is standardized prior to cross-validation.
This is described in the ``Materials and Methods" section of the supplementary:
 
``Predictor variables were then centered to the mean and z standardized.
[...] 100 iterations of ten-fold cross validation was performed
by randomly withholding 10\% of the dataset and training logistic regression models using the remaining data."

\paragraph{Learning and attention reveal a general relationship between population activity and behavior \cite{Ni2018}}
In this study PCA was performed on the entire data set. Then the first principal axis was extracted and a cross-validated
predictor was trained
using projections on this principal axis.
This is described in the left column of the 2\textsuperscript{nd} page of the article: 

``We performed principal component analysis (PCA) on population responses to the same repeated stimuli used to compute spike
count correlations
(fig. S3),
meaning that the first PC is by definition the axis that explains more of the correlated variability than any other dimension
[...]
A linear, cross-validated choice decoder (Fig. 4A) could detect differences in hit versus miss trial responses to the changed
stimulus from V4
population activity along the first PC alone as well as it could from our full data set"

\paragraph{Morphology, muscle capacity, skill, and maneuvering ability in hummingbirds \cite{Dakin2018}}

In this work, the authors perform a type of categorical cut-off. This is described in the ``Statistical Analysis" section
of the supplementary
to their paper:

``We ran a cross-validation procedure for each discriminant analysis to evaluate how well it could categorize the species.
We first fit the discriminant model using a partially random subset of 72\% of the complete records ($n = 129$ out of 180
individuals),
and then used the resulting model to predict the species labels for the remaining 51 samples.
The subset for model building included all individuals from species with fewer than 3 individuals and randomly selected 2/3
of the individuals
from species with $\ge 3$ individuals."

\paragraph{Detection and localization of surgically resectable cancers with a multi-analyte blood test \cite{Cohen2018}}
Mutant Allele Frequency normalization was performed on the entire data set prior to cross-validation.
This is described in page 2 of their supplementary:

``1) MAF normalization. All mutations that did not have $>1$ supermutant in at least one well were excluded from the analysis.
The mutant allele frequency (MAF), defined as the ratio between the total number of supermutants in each well from that sample
and the total
number of UIDs in the same well from that sample,
was first normalized based on the observed MAFs for each mutation in a set of normal controls comprising the normal plasmas
in the training set
plus a set of 256 WBCs from unrelated healthy individuals.
All MAFs with $<100$ UIDs were set to zero.
[...]
Standard normalization, i.e. subtracting the mean and dividing by the standard deviation, did not perform as well in cross-validation.
"

\paragraph{Predicting reaction performance in C-N cross-coupling using machine learning \cite{Ahneman2018}}
Feature standardization was performed prior to cross-validation. The authors specifically describe this choice in the ``Modeling"
section of
the supplementary
and mention the more conservative choice of learning the rescaling parameters only on the training set:

``The descriptor data was centered and scaled prior to data-splitting and modeling using the \software{scale(x)} function
in \software{R}.
This function normalizes the descriptors by substracting the mean and dividing by the standard deviation. An alternative
approach to feature
normalization, not used in this study,
involves scaling the training set and applying the mean and variance to scale the test set."

\paragraph{Ancient human parallel lineages within North America contributed to a coastal expansion \cite{Scheib2018}}
In this work, the authors filtered out all of the genetic variants with Minor Allele Frequency (MAF) below 5\% and later
performed cross-validation
on the filtered data set.
This is described in the ``ADMIXTURE analysis" section of their supplementary:

``the worldwide comparative dataset [...] was pruned for \software{---maf 0.05} [...]
and run through \software{ADMIXTURE v 1.23} in 100 independent runs with default settings plus \software{---cv} to identify
the 5-fold cross-validation
error at each k".

\vskip 0.2in
\small

\bibliographystyle{chicago}
\bibliography{unsupervised-preprocessing}

\end{document}